\documentclass[12pt,reqno]{amsart}

\usepackage{amssymb,amsthm,amsmath,mathrsfs,}
\usepackage{enumerate}
\usepackage{fullpage}
\usepackage{subcaption}

\numberwithin{equation}{section}

\begin{document}

\newtheorem{theorem}{Theorem}[section]
\newtheorem{lemma}[theorem]{Lemma}
\newtheorem{define}[theorem]{Definition}
\newtheorem{remark}[theorem]{Remark}
\newtheorem{corollary}[theorem]{Corollary}
\newtheorem{example}[theorem]{Example}
\newtheorem{assumption}[theorem]{Assumption}
\newtheorem{proposition}[theorem]{Proposition}
\newtheorem{conjecture}[theorem]{Conjecture}

\def\Ref#1{Ref.~\cite{#1}}

\def\i{\mathrm{i}}
\def\Rnum{{\mathbb R}}
\def\const{\text{const.}}

\def\smallbinom#1#2{{\textstyle \binom{#1}{#2}}}
\def\smallsum{\textstyle\sum}

\def\pr{{\rm pr}}
\def\rk{{\rm rank}}
\def\spn{{\rm span}}
\def\dom{{\rm dom}}
\def\ran{{\rm ran}}
\def\coker{{\rm coker}}
\def\hook{\rfloor\,}
\def\smallsum{\textstyle\sum}

\def\X{\mathbf{X}}
\def\Y{\mathbf{Y}}
\def\w{\boldsymbol\omega}
\def\d{\mathrm{d}}
\def\id{\mathrm{id}}
\def\lieder#1{{\mathcal L}_{#1}}
\def\ad{\mathrm{ad}}
\def\t{\mathrm{t}}
\def\div{\text{div}}

\def\scal{{\rm scal.}}
\def\trans{{\rm trans.}}
\def\can{{\rm can.}}
\def\p{{\rm p}}

\def\PQpair#1#2{\langle #1,#2\rangle}

\def\Jsp{\mathrm{J}}
\def\Esp{\mathcal{E}}
\def\symmsp{\mathrm{Symm}}
\def\adjsymmsp{\mathrm{AdjSymm}}
\def\multrsp{\mathrm{Multr}}

\def\Rop{{\mathcal R}}
\def\Dop{{\mathcal D}}
\def\Hop{{\mathcal H}}
\def\Jop{{\mathcal J}}

\def\F{{F}}

\def\d{{\mathbf d}}
\def\grad{{\boldsymbol\nabla}}
\def\lapl{{\boldsymbol\Delta}}

\tolerance=50000
\allowdisplaybreaks[3]

\title{Symmetry actions and brackets\\ for adjoint-symmetries.\\ I: Main results and applications}

\author{
Stephen C. Anco${}^\dagger$
\\\\\scshape{
D\lowercase{\scshape{epartment}} \lowercase{\scshape{of}} M\lowercase{\scshape{athematics and}} S\lowercase{\scshape{tatistics}}\\
B\lowercase{\scshape{rock}} U\lowercase{\scshape{niversity}}\\
S\lowercase{\scshape{t.}} C\lowercase{\scshape{atharines}}, ON L2S3A1, C\lowercase{\scshape{anada}}
}}

\thanks{${}^\dagger$sanco@brocku.ca}

\begin{abstract}
Infinitesimal symmetries of a partial differential equation (PDE) can be defined 
algebraically as the solutions of the linearization (Frechet derivative) equation 
holding on the space of solutions to the PDE, 
and they are well-known to comprise a linear space having the structure of a Lie algebra.
Solutions of the adjoint linearization equation 
holding on the space of solutions to the PDE
are called adjoint-symmetries. 
Their algebraic structure for general PDE systems is studied herein. 
This is motivated by the correspondence between 
variational symmetries and conservation laws arising from Noether's theorem, 
which has a modern generalization to non-variational PDEs, 
where infinitesimal symmetries are replaced by adjoint-symmetries,
and variational symmetries are replaced by multipliers 
(adjoint-symmetries satisfying a certain Euler-Lagrange condition). 
Several main results are obtained. 
Symmetries are shown to have three different linear actions 
on the linear space of adjoint-symmetries. 
These linear actions are used to construct bilinear adjoint-symmetry brackets,
one of which is a pull-back of the symmetry commutator bracket
and has the properties of a Lie bracket. 
The brackets do not use or require the existence of any local variational structure (Hamiltonian or Lagrangian) and thus apply to general PDE systems. 
One of the symmetry actions is shown to encode a pre-symplectic (Noether) operator, 
which leads to the construction of symplectic 2-form and Poisson bracket for evolution systems. 
The generalized KdV equation in potential form is used to illustrate all of the results.
\end{abstract}

\maketitle

\section{Introduction}

In the study of partial differential equations (PDEs), 
symmetries \cite{Ovs-book,Olv-book,BCA-book}
are a fundamental intrinsic (coordinate-free) structure of a PDE  
and have numerous important uses, such as
finding exact solutions, mapping known solutions into new solutions, 
detecting integrability, and finding linearizing transformations. 
In addition, when a PDE has a variational principle, 
then through Noether's theorem \cite{Olv-book,BCA-book}
the infinitesimal symmetries of the PDE under which the variational principle is invariant
--- namely, variational symmetries --- 
yield conservation laws. 

Like symmetries, 
conservation laws \cite{Olv-book,BCA-book,KraVin,Anc-review}
are another important intrinsic (coordinate-free) structure of a PDE.
They provide conserved quantities and conserved norms, 
which are used in the analysis of solutions; 
they detect integrability and can be used to find linearizing transformations; 
they also can be used to check the accuracy of numerical solution methods
and give rise to discretizations with good properties. 

A modern form of the Noether correspondence between variational symmetries and conservation laws
has been developed in the past few decades \cite{Olv-book,KraVin,Anc-review,Mar-Alo,Zha86,AncBlu1997,AncBlu2002b}
and generalized to non-variational PDEs. 
From a purely algebraic viewpoint, 
infinitesimal symmetries of a PDE are the solutions of the linearization (Frechet derivative) equation
holding on the space of solutions to the PDE. 
Solutions of the adjoint linearization equation, holding on the space of solutions to the PDE, are called adjoint-symmetries \cite{AncBlu1997,SarCanCra1987,SarPriCra1990}. 
In the generalization of the Noether correspondence, 
infinitesimal symmetries are replaced by adjoint-symmetries,
and variational symmetries are replaced by multipliers 
which are adjoint-symmetries satisfying an Euler-Lagrange condition
\cite{KraVin,Anc-review,AncBlu1997,AncBlu2002b}. 
(Multipliers are alternatively known as cosymmetries in the literature on Hamiltonian and integrable systems \cite{Bla-book,Sok-book}. 
The property of existence of an adjoint-symmetry for a PDE has been called ``nonlinear self-adjointness'' in some papers; see \Ref{Anc2017} and references therein.)

As an important consequence of the modern Noether correspondence, 
the problem of finding the conservation laws for a PDE 
is reduced to a kind of adjoint of the problem of finding the symmetries of the PDE. 
In particular, for any PDE system, 
conservation laws can be explicitly derived in a similar algorithmic way 
to the standard way that symmetries are derived
(see \Ref{Anc-review} for a review).

These developments motivate studying the basic mathematical properties of adjoint-symmetries
and their connections to infinitesimal symmetries. 
As is well known, 
the set of infinitesimal symmetries of a PDE has the structure of a Lie algebra,
in which the subset of variational symmetries is a Lie subalgebra, 
and the set of conservation laws of a PDE is mapped into itself under the symmetries of the PDE. 
This leads to several interesting basic questions:
\begin{itemize}
\item
How do symmetries act on adjoint-symmetries and multipliers? 
\item
Does the set of adjoint-symmetries have any kind of algebraic structure, 
such as a generalized Lie bracket or Poisson bracket, 
with the set of multipliers inheriting a corresponding structure? 
\item 
Do there exist generalized analogs of Hamiltonian and (Noether) symplectic operators 
for general PDE systems? 
\end{itemize}

In \Ref{Anc2016,AncKar}, 
the explicit action of infinitesimal symmetries on multipliers 
is derived for general PDE systems
and used to study invariance of conservation laws under symmetries. 
Recently in \Ref{AncWan2020b}, 
for scalar PDEs, 
a linear mapping from infinitesimal symmetries into adjoint-symmetries is constructed 
in terms of any fixed adjoint-symmetry that is not a multiplier. 
This mapping can be viewed as a (Noether) pre-symplectic operator, 
in analogy with symplectic operators that map symmetries into adjoint-symmetries 
for Hamiltonian systems \cite{Olv-book,FucFok}. 
The inverse mapping thus can be viewed as a pre-Hamiltonian operator. 

The present paper expands substantially on this work 
and will give answers to the basic questions just posed for general PDE systems. 

Firstly, 
it will be shown that there are two basic different actions of 
infinitesimal symmetries on adjoint-symmetries. 
One action represents a Lie derivative, 
and the other action comes from the adjoint relationship between the determining equation for infinitesimal symmetries and adjoint-symmetries. 
For adjoint-symmetries that are multipliers, 
these two actions coincide with the known action of symmetries on multipliers 
(see \Ref{Olv-book,Anc2016,AncKar}). 
Furthermore, 
the difference of the two actions produces a third action that vanishes on multipliers. 
This third action yields a generalization of the pre-symplectic operator for scalar PDEs,
and its inverse provides a general pre-Hamiltonian operator.  
For evolution PDEs and Euler-Lagrange PDEs, 
this structure further yields a symplectic 2-form and an associated Poisson bracket, 
which can be used to look for a corresponding Hamiltonian structure
for non-dissipative PDE systems. 

Secondly, 
these three actions of infinitesimal symmetries on adjoint-symmetries 
will be used to construct associated bracket structures 
on the subset of adjoint-symmetries given by the range of each action. 
Two different constructions will be given:
the first bracket is antisymmetric and can be viewed as a pull-back of the symmetry commutator (Lie bracket) to adjoint-symmetries; 
the second bracket is non-symmetric and does not utilize the commutator structure of symmetries. 
Most significantly, 
one of the antisymmetric brackets will be shown to satisfy the Jacobi identity,
and thus it gives a Lie algebra structure to a natural subset of adjoint-symmetries. 
In certain situations, this subset will coincide with the whole set of adjoint-symmetries. 
More generally, 
a correspondence (homomorphism) will exist between 
Lie subalgebras of symmetries and adjoint-symmetries, 
which will hold even for dissipative PDEs that lack any local variational (Hamiltonian or Lagrangian) structure. 

Thirdly, 
the Lie bracket on adjoint-symmetries induces a corresponding bracket structure 
for conservation laws, which is a broad generalization of a Poisson bracket 
applicable to non-Hamiltonian systems. 

All of these main results are new and provide important steps in understanding 
the basic algebraic structure of adjoint-symmetries 
and its application to pre-Hamiltonian operators, (Noether) pre-symplectic operators, 
and symplectic 2-forms for general PDE systems. 

Apart from the intrinsic mathematical interest in developing and exploring such structures, 
a more applied utilization of the results is that 
symmetry actions on adjoint-symmetries 
can be used to produce a new adjoint-symmetry  --- and hence possibly a multiplier ---
from a known adjoint-symmetry and a known symmetry, 
while brackets on adjoint-symmetries 
allow a pair of known adjoint-symmetries to generate a new adjoint-symmetry 
--- and hence possibly a multiplier --- 
just as a pair of known symmetries can generate a new symmetry
from their Lie bracket. 
Additional adjoint-symmetries can be obtained through the interplay of these structures.

The main results will be illustrated by using 
the generalized Korteweg-de Vries (gKdV) equation in potential form as running example.
Several physical examples of PDE systems will be considered in a sequel paper. 

The rest of the present paper is organized as follows. 
Section~\ref{sec:symms.adjointsymms} gives a short review of 
infinitesimal symmetries, adjoint-symmetries, and multipliers,
from an algebraic viewpoint. 
Section~\ref{sec:symmaction} presents the actions of infinitesimal symmetries on adjoint-symmetries and multipliers. 
Section~\ref{sec:noetherops} explains the construction of general pre-Hamiltonian and pre-symplectic (Noether) operators from these actions. 
Section~\ref{sec:adjsymmbracket} derives the bracket structures for adjoint-symmetries, 
and discusses their properties. 
In particular, 
the conditions under which a Lie algebra structure arises for adjoint-symmetries 
from a commutator bracket is explained. 
Section~\ref{sec:evolPDEs}
specializes the results to evolution PDEs. 
Construction of a pre-symplectic operator and an associated symplectic 2-form and Poisson bracket is also explained. 
Finally, section~\ref{sec:remarks} provides some concluding remarks. 

Throughout, the mathematical setting will be calculus in jet space \cite{Olv-book},
which is summarized in an Appendix.
Partial derivatives and total derivatives will be denoted using
a standard (multi-) index notation. 
The Frechet derivative will be denoted by ${}'$. 
Adjoints of total derivatives and linear operators will be denoted by ${}^*$.
Prolongations will be denoted as $\pr$. 

Hereafter,
a ``symmetry'' will refer to an infinitesimal symmetry in evolutionary form. 

Work on classifying adjoint-symmetries of PDEs can be found in 
\Ref{AncBlu1997,Anc2017,AncBlu2002a,Lun2008,ZhaXie2016,Zha2017,Ma2018}. 
See also \Ref{FelYas} for other recent work on symplectic operators and variational structure related to adjoint-symmetries from a cohomological perspective.

\section{Symmetries and adjoint-symmetries}\label{sec:symms.adjointsymms}

An algebraic perspective will be utilized to allow 
symmetries and adjoint-symmetries to be defined and handled in a unified way
(following \Ref{Anc-review}). 

Consider a general PDE system of order $N$ consisting of $M$ equations
\begin{equation}\label{pde.sys}
G^A(x,u^{(N)}) =0,
\quad
A=1,\ldots,M
\end{equation}
where $x^i$, $i=1,\ldots,n$, are the independent variables,
and $u^\alpha$, $\alpha=1,\ldots,m$, are the dependent variables.
The space of formal solutions $u^\alpha(x)$ of the PDE system will be denoted $\Esp$.
As is usual in symmetry theory \cite{Ovs-book,Olv-book,BCA-book},
the PDE system will be assumed to be well posed in the sense that the standard tools of variational calculus in jet space can be applied. 
In particular, 
no integrability conditions are assumed to arise from the equations and their differential consequences; 
namely, the PDE system and its differential consequences are involutive. 
(A more precise formulation can be found 
in \cite{Olv-book,KraVin,Vin1998} from a geometric/algebraic point of view, 
and in \cite{Anc-review} from a computational point of view.
A general reference on involutivity, which bridges these viewpoints, is \cite{Sei-book}.)

An underlying technical condition will be that a PDE system admits a solved-form for a set of leading derivatives,
and likewise all differential consequences of the PDE system admit a solved-form in terms of differential consequences of the leading derivatives. 
This condition allows Hadamard's lemma to hold in the setting of jet space \cite{Anc-review}. 

\begin{lemma}\label{lem:Hadamard}
If a function $f(x,u^{(k)})$ vanishes on $\Esp$ then $f=R_f(G)$ holds identically,
where $R_f$ is some linear differential operator in total derivatives 
whose coefficients are functions that are non-singular on $\Esp$. 
\end{lemma}

When the preceding technical conditions hold, a PDE system will be called \emph{regular}. 
Essentially all PDE systems of interest in physical applications are regular systems. 
(See \Ref{Anc-review,AncChe} for examples and further discussion.)
Hereafter, only regular PDE systems are considered. 

An additional technical condition, which is not needed for the main results, will be useful for certain developments. 
The proof is similar to that of the previous lemma \cite{Anc-review}. 

\begin{lemma}\label{lem:no.diff.identities}
Suppose $R(G)=0$ holds identically for a linear differential operator $R$ in total derivatives whose coefficients are functions that are non-singular on $\Esp$. 
If the PDE system $G^A=0$ does not obey any differential identities, 
then $R$ vanishes on $\Esp$. 
\end{lemma}

For a running example, the focusing gKdV equation in potential form will be used: 
\begin{equation}\label{p-gkdv}
u_t + \tfrac{1}{p+1} (u_x)^{p+1} + u_{xxx} =0 
\end{equation}
where $p>0$ is an arbitrary nonlinearity power. 
This equation will be referred to as the (focusing) p-gKdV equation. 
It is a regular PDE system. 
Note that its $x$-derivative yields the focusing gKdV equation in physical form 
$v_t + v^p v_x + v_{xxx}=0$
with $v=u_x$, 
where the coefficients of the convective dispersion terms are scaled to $1$. 
The special cases $p=1,2$ are the KdV equation and the mKdV equation, 
which are integrable systems.

\subsection{Determining equations and identities}

An infinitesimal \emph{symmetry} of a PDE system \eqref{pde.sys} is a set of functions 
$P^\alpha(x,u^{(k)})$ that are non-singular on $\Esp$ and satisfy 
\begin{equation}\label{symm.deteqn}
G'(P)^A|_\Esp =0 . 
\end{equation}
This is the determining equation for $P^\alpha$,
called the characteristic functions of the symmetry. 

Off of the solution space $\Esp$,
the symmetry determining equation is given by
\begin{equation}\label{symm.deteqn.offsoln}
G'(P)^A = R_P(G)^A
\end{equation}
(due to Lemma~\ref{lem:Hadamard}) 
where $R_P=(R_P)^{A\,I}_{B} D_I$ is some linear differential operator in total derivatives 
whose coefficients $(R_P)^{A\,I}_{B}$ are functions that are non-singular on $\Esp$. 

The determining equation for adjoint-symmetries is the adjoint of 
the symmetry determining equation \eqref{symm.deteqn}. 
It is obtained by using the Frechet derivative identity
\begin{equation}\label{symm.adjsymm.ibp}
Q_A G'(P)^A = P^\alpha G'{}^*(Q)_\alpha + D_i\Psi^i(P,Q) . 
\end{equation}  
There is an explicit expression for $\Psi^i$ in terms of $G^A$ 
(see \cite{Anc-review} and references therein). 

An \emph{adjoint-symmetry} of a PDE system \eqref{pde.sys} 
is a set of functions $Q_A(x,u^{(k)})$ that are non-singular on $\Esp$ and satisfy
\begin{equation}\label{adjsymm.deteqn}
G'{}^*(Q)_\alpha|_\Esp =0 . 
\end{equation}
Off of the solution space $\Esp$, this determining equation is given by
\begin{equation}\label{adjsymm.deteqn.offsoln}
G'{}^*(Q)_\alpha = R_Q(G)_\alpha
\end{equation}
(again due to Lemma~\ref{lem:Hadamard}) 
where $R_Q=(R_Q)_{\alpha\,B}^{I} D_I$ is some linear differential operator in total derivatives
whose coefficients $(R_Q)_{\alpha\,B}^{I}$ are functions that are non-singular on $\Esp$.

The geometrical meaning of symmetries is well known.
From the algebraic viewpoint, it comes from the relation 
$G'(P){}^A = (\pr P^\alpha\partial_{u^\alpha}) G^A$
whereby the symmetry determining equation \eqref{symm.deteqn}
can be expressed as
\begin{equation}\label{symm.deteqn.vector}
((\pr P^\alpha\partial_{u^\alpha})G^A)|_\Esp =0 . 
\end{equation}
This is usually the starting point for defining symmetries,
since it indicates that $\X_P = P^\alpha\partial_{u^\alpha}$
is a vector field that is tangent to surfaces $G^A=0$ 
(and their prolongations $D^k G^A=0$, $k=0,1,2,\ldots$) in jet space. 
A geometrical meaning for adjoint-symmetries has recently been developed in \Ref{AncWan2020a}, based on evolutionary 1-forms $Q_A\d G^A$ 
that functionally vanish on the solution space $\Esp$. 

The most common form encountered for symmetries is a Lie point symmetry \cite{Olv-book,BCA-book}, 
given by $P^\alpha = \eta^\alpha(x,u) -\xi^i(x,u) u^\alpha_i$. 
Symmetries that have have a general form $P^\alpha(x,u^{(k)})$ with $k\geq 1$
are sometimes called generalized symmetries or symmetries of order $k$. 

The most common form for adjoint-symmetries is given by 
$Q_A(x,u^{(k)})$ with $k<N$, 
where $N$ is the differential order of a given PDE system \eqref{pde.sys}. 
Such adjoint-symmetries are called \emph{low-order} \cite{Anc-review,AncKar}. 

A symmetry or an adjoint-symmetry is called higher-order if has a differential order $k>N$. 
Existence of an infinite hierarchy with $k$ being unbounded is typically an indicator of integrability \cite{Olv-book,MikShaSok-book}. 

\underline{Running example}: 
For the p-gKdV equation \eqref{p-gkdv}, 
the symmetry and adjoint-symmetry determining equations are respectively given by 
\begin{equation}
(D_t P + u_x^p D_x P + D_x^3 P)|_\Esp =0,
\quad
(-D_t Q -D_x(u_x^p Q) - D_x^3 Q)|_\Esp =0 .
\end{equation}
These equations are adjoints of each other. 
Since they do not coincide when $P=Q$ with $p\neq0$, 
this shows that p-gKdV adjoint-symmetries differ from p-gKdV symmetries. 
It is well-known that, for arbitrary $p>0$, 
the Lie point symmetries are spanned by 
\begin{equation}\label{p-gkdv.P}
P_1=1,
\quad
P_2=-u_x,
\quad
P_3=-u_t,
\quad
P_4=(p-2)u -3p t u_t -xp u_x,
\end{equation}
which respectively generate shifts, space-translations, time-translations, and scalings. 
They satisfy 
\begin{equation}\label{p-gkdv.RP}
R_{P_1} = 0,
\quad
R_{P_2} = -D_x,
\quad
R_{P_3} = -D_t,
\quad
R_{P_4} = -p x D_x -3p t D_t - 2(p+1) 
\end{equation}
off of $\Esp$. 
The low-order adjoint-symmetries can be shown to be spanned by 
\begin{equation}\label{p-gkdv.Q}
Q_1=u_{xx},
\quad
Q_2=u_{tx},
\quad
Q_3=2u_x + 3p t u_{tx} + p x u_{xx} 
\end{equation}
where
\begin{equation}\label{p-gkdv.RQ}
R_{Q_1} =-D_x^2, 
\quad
R_{Q_2} = -D_tD_x,
\quad
R_{Q_3} = -p x D_x^2 -3p t D_tD_x - (3p+2)D_x . 
\end{equation}
In the special cases $p=1,2$, 
a hierarchy of higher-order symmetries and adjoint-symmetries exist, 
corresponding to the integrability structure of the KdV and mKdV equations. 
(No integrability structure is known for any other values of $p\neq0$.)

Recall that a \emph{multiplier} is a set of functions $\Lambda_A(x,u^{(k)})$
that are non-singular on $\Esp$ and satisfy
$\Lambda_A G^A = D_i\Psi^i$ off of $\Esp$, 
for some vector function $\Psi^i$ in jet space. 
This total divergence condition is equivalent to
\begin{equation}\label{multr.deteqn}
E_{u^\alpha}(\Lambda_A G^A)=0 . 
\end{equation}
It can be further reformulated through the product rule of the Euler operator,
which yields the equivalent condition 
$\Lambda'{}^*(G)_\alpha + G'{}^*(\Lambda)_\alpha =0$. 
Consequently, on $\Esp$, 
\begin{equation}
G'{}^*(\Lambda)_\alpha|_\Esp =0
\end{equation}
whereby $\Lambda_A$ is an adjoint-symmetry. 
Off of $\Esp$, 
the adjoint-symmetry determining equation \eqref{adjsymm.deteqn.offsoln}
yields 
\begin{equation}
G'{}^*(\Lambda)_\alpha = R_\Lambda(G)_\alpha
\end{equation}
where $R_\Lambda$ is a linear differential operator in total derivatives. 
Hence, one sees that 
$\Lambda'{}^*(G)_\alpha =- G'{}^*(\Lambda)_\alpha =-R_\Lambda(G)_\alpha$.
Now suppose that $G^A=0$ does not obey any differential identities. 
Then one can conclude (from Lemma~\ref{lem:no.diff.identities}) that 
$\Lambda'{}^* =-R_\Lambda + S^{I,J} (D_I G) D_J$
where $S^{I,J}=-S^{J,I}$ holds off of $\Esp$ and $S^{I,J}$ is non-singular on $\Esp$.
Furthermore, suppose that $\Lambda_A$ contains 
no leading derivatives of $G^A=0$ and no differential consequences of any leading derivatives.
Then one can assume without loss of generality that $S^{I,J}=0$.
Therefore, in this situation,
$\Lambda'{}^*=-R_\Lambda$
holds identically.
The adjoint of this equation yields the relation
\begin{equation}\label{multr.relation}
\Lambda'{}=-R_\Lambda^* . 
\end{equation}

Every multiplier $\Lambda_A(x,u^{(k)})$ of a PDE system determines 
a conservation law $(D_i\Psi^i)|_\Esp=0$ holding on the solution space $\Esp$. 
The components $\Psi^i$ can obtained from $\Lambda_A$ by 
homotopy integral formulas \cite{Olv-book,BCA-book,Anc-review}, 
or by an algebraic formula when the given PDE system possesses 
a scaling symmetry \cite{Anc-review,Anc2016}. 
When a PDE system is regular, 
all conservation laws will arise from multipliers \cite{Anc-review}. 

\underline{Running example}: 
The low-order multipliers of the p-gKdV equation \eqref{p-gkdv} consist of 
the span of a subset of the low-order adjoint-symmetries:
\begin{equation}\label{p-gkdv.multr}
\Lambda_1 = Q_1 = u_{xx},
\quad
\Lambda_2 =Q_2 = u_{tx} . 
\end{equation}
In particular, 
the adjoint-symmetry $Q_3=2u_x + 3p t u_{t,x} + p x u_{xx}$ is not a multiplier. 
The conservation laws arising from the two multipliers are respectively given by
\begin{equation}\label{p-gkdv.momentum}
(\Psi^t,\Psi^x)=(-\tfrac{1}{2} u_x{}^2, \tfrac{1}{2}u_{xx}{}^2 + u_t u_x +\tfrac{1}{(p+1)(p+2)} u_x{}^{p+1})
\end{equation}
and 
\begin{equation}\label{p-gkdv.energy}
(\Psi^t,\Psi^x)=(-\tfrac{1}{2}u_{xx}{}^2 + \tfrac{1}{(p+1)(p+2)} u_x{}^{p+2},u_{tx}u_{xx}+\tfrac{1}{2} u_x{}^2) .
\end{equation}
These describe continuity equations for momentum and energy, 
which can be seen from the form of the conserved densities 
$\Psi^t=\tfrac{1}{2}v^2,\tfrac{1}{2}v_x{}^2 -\tfrac{1}{(p+1)(p+2)} v{}^{p+2}$
(up to an overall sign)
expressed in terms of the gKdV variable $v=u_x$
(see e.g.\  \Ref{AncNayRec}).

\section{Action of symmetries on adjoint-symmetries}\label{sec:symmaction}

Symmetries of any given PDE system 
are well-known to form a Lie algebra via their commutators. 
From the algebraic viewpoint,
if $P_1^\alpha$, $P_2^\alpha$ are symmetries,
then so is the commutator defined by 
\begin{equation}\label{symm.commutator}
[P_1,P_2]^\alpha = P_2{}'(P_1)^\alpha - P_1{}'(P_2)^\alpha . 
\end{equation}
The geometrical formulation is the same:
\begin{equation}
[\pr\X_{P_1},\pr\X_{P_2}] = \pr\X_{[P_1,P_2]} . 
\end{equation}
Stated precisely, the set of symmetries comprises a linear space
on which the commutator defines a bilinear antisymmetric bracket that obeys the Jacobi identity. 
This bracket is called the \emph{Lie bracket} of the symmetry vector fields. 
Any symmetry has a natural action on the linear space of all symmetries 
via the algebraic commutator \eqref{symm.commutator}. 
This action is commonly denoted by $\ad(P_1)P_2 = [P_1,P_2]$. 

Symmetries also have a natural action on the set of adjoint-symmetries,
since this set is a linear space that is determined by the given PDE system
whose solution set $\Esp$ is mapped into itself by a symmetry. 
Actually, 
there are two distinct actions of symmetries on the linear space of adjoint-symmetries,
as shown next. 

The first symmetry action arises directly from the prolonged action of a symmetry $P^\alpha$ 
applied to the adjoint-symmetry determining equation \eqref{adjsymm.deteqn.offsoln}.
To begin, from the lefthand side of this equation, one gets
\begin{equation}\label{eqn.lhs}
\pr\X_P(G'{}^*(Q)_\alpha)
=G'{}^*(\pr\X_P(Q))_\alpha + \pr\X_P(G'{}^*)(Q)_\alpha . 
\end{equation}
The last term can be simplified by the following steps.
First, one has 
$\pr\X_P(G'{}^*) = (\pr\X_P(G))'{}^* -P'{}^* G'{}^*$
(by identity \eqref{X.Frechet.id}),
whence 
$\pr\X_P(G'{}^*)(Q)_\alpha = (\pr\X_P(G))'{}^*(Q)_\alpha  -P'{}^*(G'{}^*(Q))_\alpha$. 
Second, through the symmetry equation \eqref{symm.deteqn.offsoln},
one can simplify
$(\pr\X_P(G))'{}^*|_\Esp= (R_P(G))'{}^*|_\Esp = (R_P G')^*|_\Esp = G'{}^*R_P^*|_\Esp$,
where $R_P^*$ is the adjoint of the linear differential operator $R_P$ (in total derivatives). 
Thus, expression \eqref{eqn.lhs} on $\Esp$ becomes
\begin{equation}\label{eqn.lhs.simp}
\pr\X_P(G'{}^*(Q)_\alpha)|_\Esp
=G'{}^*(Q'(P) + R_P^*(Q))_\alpha|_\Esp . 
\end{equation}
Next, from the righthand side of equation \eqref{adjsymm.deteqn.offsoln}, one has
\begin{equation}\label{eqn.rhs}
\pr\X_P(R_Q(G)_\alpha) 
= (\pr\X_P R_Q)(G)_\alpha + R_Q(\pr\X_P(G))_\alpha . 
\end{equation}
On $\Esp$, this yields
\begin{equation}\label{eqn.rhs.simp}
\pr\X_P(R_Q(G)_\alpha)|_\Esp =0 . 
\end{equation}
Finally, 
from equating expressions \eqref{eqn.rhs.simp} and \eqref{eqn.lhs.simp},
one gets
\begin{equation}\label{1form.eqn.simp}
G'{}^*(Q'(P) + R_P^*(Q))_\alpha|_\Esp =0
\end{equation}
which shows that $Q'(P)_A + R_P^*(Q)_A$ is an adjoint-symmetry.
Therefore, this yields a linear mapping
\begin{equation}\label{symmaction2.adjsymm}
Q_A\overset{{\X_P}}{\longrightarrow} Q'(P)_A + R_P^*(Q)_A
\end{equation}
acting on the linear space of adjoint-symmetries.

This action \eqref{symmaction2.adjsymm} 
can be interpreted geometrically as a Lie derivative \cite{AncWan2020a}
and is a generalization of a better known action of 
symmetries on conservation law multipliers,
which is found in \Ref{Anc2016,AncKar}. 
Further discussion is given in section~\ref{sec:symmaction.multr}. 

The second symmetry action arises from the adjoint relation between
the respective determining equations \eqref{symm.deteqn} and \eqref{adjsymm.deteqn}
for symmetries and adjoint-symmetries. 

As is well known \cite{AncBlu1997,Anc2017,Cav,Lun}, 
when $P^\alpha$ is a symmetry and $Q_A$ is an adjoint-symmetry,
the adjoint relation \eqref{symm.adjsymm.ibp} yields a conservation law since
\begin{equation}\label{symm.adjsymm.conslaw}
D_i\Psi^i(P,Q)|_\Esp = Q_A G'(P)^A|_\Esp - P^\alpha G'{}^*(Q)_\alpha|_\Esp
= 0 
\end{equation}  
from the determining equations \eqref{symm.deteqn} and \eqref{adjsymm.deteqn}.
Off of $\Esp$,
this formula is given by 
$D_i\Psi^i(P,Q) = Q_A R_P(G)^A - P^\alpha R_Q(G)_\alpha$
where $R_P$ and $R_Q$ are the linear differential operators (in total derivatives)
determined by equations \eqref{symm.deteqn.offsoln} and \eqref{adjsymm.deteqn.offsoln}.
Integration by parts yields
\begin{equation}\label{eqn.symmaction1}
D_i\Psi^i(P,Q) = (R_P^*(Q)_A - R_Q^*(P)_A)G^A + D_i F^i(P,Q;G)
\end{equation}
and hence $(R_P^*(Q)_A - R_Q^*(P)_A)G^A$ is a total divergence in jet space. 
This implies that the set of functions $R_P^*(Q)_A - R_Q^*(P)_A$ constitute
a conservation law multiplier.
Since every multiplier is an adjoint-symmetry,
there is a linear mapping
\begin{equation}\label{symmaction1.adjsymm}
Q_A\overset{{\X_P}}{\longrightarrow} R_P^*(Q)_A - R_Q^*(P)_A := \Lambda_A
\end{equation}
which acts on the linear space of adjoint-symmetries. 

The preceding results are a full and complete generalization of the symmetry actions 
derived for scalar PDEs in \Ref{AncWan2020b}. 
They will now be summarized, 
and then some of their consequences will be developed. 

\begin{theorem}\label{thm:symmactions.adjsymm}
For any (regular) PDE system \eqref{pde.sys}, 
there are two actions \eqref{symmaction2.adjsymm} and \eqref{symmaction1.adjsymm}
of symmetries on the linear space of adjoint-symmetries.
The second symmetry action \eqref{symmaction1.adjsymm} maps adjoint-symmetries
into conservation law multipliers. 
The difference of the first and second actions yields the linear mapping
\begin{equation}\label{symmaction3.adjsymm}
Q_A\overset{{\X_P}}{\longrightarrow} Q'(P)_A + R_Q^*(P)_A . 
\end{equation}
\end{theorem}

The action \eqref{symmaction3.adjsymm} will be trivial
when the adjoint-symmetry is a conservation law multiplier, 
as follows from the relation \eqref{multr.relation} which holds 
under certain mild conditions on the form of 
the PDE system $G^A=0$ (Lemma~\ref{lem:no.diff.identities}) 
and the functions $Q_A$. 

\begin{proposition}\label{prop:symmaction3.adjsymm}
For a (regular) PDE system $G^A=0$ with no differential identities,
the symmetry action \eqref{symmaction3.adjsymm} on 
adjoint-symmetries $Q_A$ that contain no leading derivatives (and their differential consequences) in the PDE system
is trivial iff $Q_A$ is a conservation law multiplier. 
\end{proposition}

The conditions in Proposition~\ref{prop:symmaction3.adjsymm}
are satisfied by evolution PDEs,
as shown in section~\ref{sec:evolPDEs}. 

\underline{Running example}: 
For the p-gKdV equation \eqref{p-gkdv}, 
the symmetry actions on adjoint-symmetries are shown in Table~\ref{table:p-gkdv.symmactions}. 
The non-zero commutators of the symmetries are given by 
\begin{equation}\label{p-gkdv.symmalg}
[P_1,P_4]=(p-2)P_1,
\quad
[P_2,P_4]=pP_2,
\quad
[P_3,P_4]=3pP_3 .
\end{equation}

\begin{table}[h!]
\caption{p-gKdV equation: symmetry actions on adjoint-symmetries}
\label{table:p-gkdv.symmactions}
\begin{subtable}{.5\linewidth}
\centering
\caption{action by \eqref{symmaction2.adjsymm}} 
\begin{tabular}{l||c|c|c|c}
& $P_1$
& $P_2$
& $P_3$
& $P_4$
\\
\hline
\hline  
$Q_1$
& $0$
& $0$
& $0$
& $(p-4)Q_1$
\\
\hline
$Q_2$
& $0$
& $0$
& $0$
& $-(p+4)Q_2$
\\
\hline
$Q_3$
& $0$
& $p Q_1$
& $3p Q_2$
& $2(p-2) Q_3$
\\
\end{tabular}
\end{subtable}%
\begin{subtable}{.5\linewidth}
\centering
\caption{action by \eqref{symmaction1.adjsymm}} 
\begin{tabular}{l||c|c|c|c}
& $P_1$
& $P_2$
& $P_3$
& $P_4$
\\
\hline
\hline  
$Q_1$
& $0$
& $0$
& $0$
& $(p-4)Q_1$
\\
\hline
$Q_2$
& $0$
& $0$
& $0$
& $-(p+4)Q_2$
\\
\hline
$Q_3$
& $0$
& $(4-p) Q_1$
& $(p+4) Q_2$
& $0$
\\
\end{tabular}
\end{subtable}%
\\
\begin{subtable}{.5\linewidth}
\caption{action by \eqref{symmaction3.adjsymm}}
\begin{tabular}{l||c|c|c|c}
& $P_1$
& $P_2$
& $P_3$
& $P_4$
\\
\hline
\hline  
$Q_1$
& $0$
& $0$
& $0$
& $0$
\\
\hline
$Q_2$
& $0$
& $0$
& $0$
& $0$
\\
\hline
$Q_3$
& $0$
& $2(p-2) Q_1$
& $2(p-2) Q_2$
& $2(p-2) Q_3$
\\
\end{tabular}
\end{subtable}
\end{table}

\subsection{Symmetry action on multipliers}\label{sec:symmaction.multr}

The action of a symmetry vector field $\X_P = P^\alpha\partial_{u^\alpha}$
on the multiplier equation $\Lambda_A G^A = D_i\Psi^i$
yields, for the righthand side,
\begin{equation}\label{multr.eqn.rhs}
\pr\X_P D_i\Psi^i = D_i(\pr\X_P \Psi^i), 
\end{equation}
while for the lefthand side, 
$\pr\X_P(\Lambda_A G^A) = \Lambda'(P)_A G^A + \Lambda_A G'(P)^A$. 
The last term can be simplified by using the symmetry equation \eqref{symm.deteqn.offsoln} off of $\Esp$:
\begin{equation}
\Lambda_A G'(P)^A
= \Lambda_A R_P(G)^A 
=R_P^*(\Lambda)_A G^A + D_i F^i . 
\end{equation}
Thus, 
\begin{equation}\label{multr.eqn.lhs}
\pr\X_P(\Lambda_A G^A)
= (\Lambda'(P)_A + R_P^*(\Lambda)_A)G^A\quad\text{modulo total derivatives.}  
\end{equation}
Now, from equating expressions \eqref{multr.eqn.lhs} and \eqref{multr.eqn.rhs},
one concludes that
$(\Lambda'(P)_A + R_P^*(\Lambda)_A)G^A$ is a total derivative.
Therefore, $\Lambda'(P)_A + R_P^*(\Lambda)_A$ is a multiplier.

This yields the following well-known action \cite{Anc2016,AncKar}:
\begin{equation}
\Lambda_A \overset{\X_P}{\longrightarrow} \Lambda'(P)_A + R_P^*(\Lambda)_A . 
\end{equation}

Theorem~\ref{thm:symmactions.adjsymm} shows that this action extends from
conservation law multipliers to adjoint-symmetries
through the symmetry action \eqref{symmaction2.adjsymm} on adjoint-symmetries.

\subsection{Action of Lie point symmetries}\label{sec:Liepointsymm}

An explicit expression for the first symmetry action \eqref{symmaction2.adjsymm}  
in Theorem~\ref{thm:symmactions.adjsymm}
can be derived in the case of Lie point symmetries. 

A \emph{Lie point symmetry} vector field has the form \cite{Olv-book,BCA-book}
\begin{equation}\label{pointsymm}
\X_\p = P_\p^\alpha\partial_{u^\alpha} ,
\quad
P_\p^\alpha = \eta^\alpha(x,u) -\xi^i(x,u) u^\alpha_i, 
\end{equation}
which generates a point transformation group acting on the space $(x,u)$, 
as given by exponentiation of the corresponding canonical vector field 
\begin{equation}
\Y_\p = \xi^i \partial_{x^i} + \eta^\alpha\partial_{u^\alpha} .
\end{equation}
The prolongations of these vector fields are related by \cite{Olv-book,BCA-book}
\begin{equation}\label{canonical}
\pr\Y_\p = \xi^i D_i +\pr\X .
\end{equation}
A function $F(x,u^{(k)})$ is \emph{symmetry invariant} iff 
$\pr\Y_\p F$ vanishes identically. 
More generally, a function $F(x,u^{(k)})$ is \emph{symmetry homogeneous} iff 
$\pr\Y_\p F = \sigma_F F$ holds identically
for some function $\sigma_F(x,u)$. 

The symmetry determining equation \eqref{symm.deteqn.offsoln} 
for Lie point symmetries can be expressed as 
\begin{equation}
\pr\Y_\p(G)= R_\p(G)
\end{equation}
where $R_\p=(R_\p)^{A\,I}_{B} D_I$ is some linear differential operator in total derivatives 
whose coefficients $(R_\p)^{A\,I}_{B}$ are functions that are non-singular on $\Esp$. 
When every PDE in the system $G^A=0$ has the same differential order, 
and the system has no differential identities, 
then $R_\p$ will be purely algebraic, 
namely $(R_\p)^{A\,I}_{B}$ vanishes for $I\neq\emptyset$. 

\begin{proposition}\label{prop:symmaction2.pointsymm}
The first symmetry action \eqref{symmaction2.adjsymm}  
for a Lie point symmetry \eqref{pointsymm} on an adjoint-symmetry is given by 
\begin{equation}\label{pointsymm.action2.adjsymm}
Q_A \overset{{\X_\p}}{\longrightarrow} 
\Y_p(Q)_A + R_\p^*(Q)_A  +(D_i\xi^i)Q_A 
\end{equation}
where $R_\p^*$ is the adjoint of $R_\p$. 
\end{proposition}

The proof is a straightforward computation of the terms $Q'(P_\p)_A +R_{P_\p}^*(Q)_A$ 
in the action \eqref{symmaction2.adjsymm}.
One has $Q'(P_\p)_A = \pr\Y_\p(Q)_A - \xi^iD_i Q_A$ 
and $R_{P_\p}(G)^A= R_\p(G)^A - \xi^i D_i G^A$
from identity \eqref{canonical}. 
Hence, $R_{P_\p}^*(Q)_A= R_\p^*(Q)_A +D_i(\xi^i Q_A)$,
and thus after cancellation of terms, 
one obtains the action \eqref{pointsymm.action2.adjsymm}. 

Similar explicit expressions can be obtained for the other two symmetry actions \eqref{symmaction1.adjsymm}, \eqref{symmaction3.adjsymm} 
in Theorem~\ref{thm:symmactions.adjsymm}
in the case of adjoint-symmetries with a first-order linear form 
\begin{equation}\label{firstord.linear.adjsymm}
Q_A = \kappa_A(x,u) + \rho^i_{A\alpha}(x,u) u^\alpha_i . 
\end{equation}
This form is a counterpart of Lie point symmetries 
(more generally, first-order linear symmetries). 
The adjoint-symmetry determining equation \eqref{adjsymm.deteqn.offsoln} 
implies that
\begin{equation}\label{firstord.linear.R_Q}
G'{}^*(Q)_\alpha = \rho^i_{A\alpha} D_i G^A + K_{A\alpha} G^A
\end{equation}
for some functions $K_{A\alpha}$ that are non-singular on $\Esp$,
when every PDE in the system $G^A=0$ has the same differential order, 
and the system has no differential identities. 

This leads to the following result. 

\begin{proposition}\label{prop:symmaction1and3.pointsymm}
For a Lie point symmetry \eqref{pointsymm}, 
the second and third symmetry actions \eqref{symmaction1.adjsymm} and \eqref{symmaction3.adjsymm} 
on a first-order linear adjoint-symmetry \eqref{firstord.linear.adjsymm}--\eqref{firstord.linear.R_Q}
are given by 
\begin{align}
& Q_A \overset{{\X_\p}}{\longrightarrow} 
R_\p^*(Q)_A  
+ u^\alpha_j D_i( 2\xi^{[i} \rho^{j]}_{A\alpha}) 
+D_i(\xi^i\kappa_A+\rho^i_{A\alpha} \eta^\alpha) 
-K_{A\alpha}(\eta^\alpha - \xi^i u^\alpha_i), 
\label{pointsymm.action1.1stordlinear.adjsymm}
\\
& Q_A \overset{{\X_\p}}{\longrightarrow} 
\Y_\p(Q)_A +(D_i\xi^i) Q_A
-u^\alpha_j D_i(2\xi^{[i} \rho^{j]}_{A\alpha}) 
-D_i(\xi^i\kappa_A +\rho^i_{A\alpha} \eta^\alpha)  
+ K_{A\alpha}(\eta^\alpha - \xi^i u^\alpha_i) ,
\label{pointsymm.action3.1stordlinear.adjsymm}
\end{align}
where $R_\p^*$ is the adjoint of $R_\p$. 
\end{proposition}

The proof is similar to that for the action \eqref{symmaction2.adjsymm}.
One has $R_{P_\p}^*(Q)_A= R_\p^*(Q)_A +D_i(\xi^i Q_A)$,
where 
$D_i(\xi^i Q_A) = D_i(\xi^i\kappa_A) + D_i(\xi^i\rho^j_{A\alpha}) u^\alpha_j 
+ \xi^i\rho^j_{A\alpha} u^\alpha_{ij}$. 
Next, from relation \eqref{firstord.linear.R_Q}, 
one obtains 
$R_Q^*(P_\p)_A = K_{A\alpha} P_\p^\alpha -D_i(\rho^i_{A\alpha} P_\p^\alpha)$
where 
$D_i(\rho^i_{A\alpha} P_\p^\alpha) = 
D_i(\rho^i_{A\alpha} \eta^\alpha)  - D_i(\rho^i_{A\alpha} \xi^j) u^\alpha_j 
- \rho^i_{A\alpha} \xi^j u^\alpha_{ij}$
and 
$K_{A\alpha} P_\p^\alpha = K_{A\alpha}(\eta^\alpha - \xi^i u^\alpha_i)$. 
Then, combining the terms $R_{P_\p}^*(Q)_A -R_Q^*(P_\p)_A$, 
one gets expression \eqref{pointsymm.action1.1stordlinear.adjsymm}. 
Likewise, 
combining the terms $Q'(P_\p)_A+ R_Q^*(P_\p)_A$
yields expression \eqref{pointsymm.action3.1stordlinear.adjsymm}. 

Two basic types of Lie point symmetries which appear in numerous applications 
are translations $\Y_\trans = a^i\partial_{x^i}$
and scalings $\Y_\scal = w_{(i)} x^i \partial_{x^i} + w_{(\alpha)} u^\alpha \partial_{u^\alpha}$.
Here the vector $a^i$ 
represents the direction of the translation; 
the scalars $w_{(\alpha)}, w_{(i)}$ 
represent the scaling weights of $u^\alpha$ and $x^i$. 
The corresponding evolutionary form of these symmetries is given by 
\begin{equation}\label{translation.symm}
P_\trans^\alpha = -a^i u^\alpha_i
\end{equation}
and
\begin{equation}\label{scaling.symm}
P_\scal = w^{(\alpha)} u^\alpha -w^{(i)} x^i u^\alpha_i . 
\end{equation}
Their action on adjoint-symmetries has a very simple form,
which is an immediate consequence of 
Propositions~\ref{prop:symmaction2.pointsymm} and~\ref{prop:symmaction1and3.pointsymm}. 

\begin{corollary}\label{cor:symmaction.translations.scalings}
(i) Suppose $Q_A$ and $G^A$ are translation invariant: 
$\Y_\trans(Q)_A=0$ and $\Y_\trans(G)^A=0$. 
Then the three symmetry actions respectively consist of 
\begin{align}
& Q_A \overset{{\X_\p}}{\longrightarrow} 
0 ,
\label{translation.symmaction2.adjsymm}
\\
& Q_A \overset{{\X_\p}}{\longrightarrow} 
2 u^\alpha_j a^{[i} D_i \rho^{j]}_{A\alpha}
+ a^i D_i \kappa_A 
+a^i u^\alpha_i  K_{A\alpha} , 
\label{translation.symmaction1.1stordlinear.adjsymm}
\\
& Q_A \overset{{\X_\p}}{\longrightarrow} 
-2 u^\alpha_j a^{[i} D_i \rho^{j]}_{A\alpha}
-a^i D_i \kappa_A
- a^i u^\alpha_i K_{A\alpha} . 
\label{translation.symmaction3.1stordlinear.adjsymm}
\end{align}
(ii) Suppose $Q_A$ and $G^A$ are scaling homogeneous: 
$\Y_\scal(Q)_A=w^{(A)} Q_A$ and $\Y_\scal(G)^A=\omega^{(A)} G^A$. 
Then the three symmetry actions respectively consist of 
\begin{align}
&\begin{aligned}
Q_A \overset{{\X_\p}}{\longrightarrow}\ & 
(\omega^{(A)} + w^{(A)} + \smallsum_i w^{(i)}) Q_A ,
\end{aligned}
\label{scaling.symmaction2.adjsymm}
\\
&\begin{aligned}
Q_A \overset{{\X_\p}}{\longrightarrow}\ & 
\omega^{(A)} Q_A  
+ u^\alpha_j w^{(i)} D_i( 2x^{[i} \rho^{j]}_{A\alpha}) 
+ w^{(i)} D_i(x^i\kappa_A)+ w^{(\alpha)} D_i(\rho^i_{A\alpha} u^\alpha) 
\\&\qquad
-K_{A\alpha}(w^{(\alpha)} u^\alpha - w^{(i)} x^i u^\alpha_i), 
\end{aligned}
\label{scaling.symmaction1.1stordlinear.adjsymm}
\\
& \begin{aligned}
Q_A \overset{{\X_\p}}{\longrightarrow}\ & 
(w^{(A)} + \smallsum_i w^{(i)}) Q_A
-u^\alpha_j w^{(i)} D_i(2x^{[i} \rho^{j]}_{A\alpha})
- w^{(i)} D_i(x^i\kappa_A) + w^{(\alpha)} D_i(\rho^i_{A\alpha} u^\alpha)  
\\&\qquad
+ K_{A\alpha}(w^{(\alpha)} u^\alpha - w^{(i)} xi^i u^\alpha_i) . 
\end{aligned}
\label{scaling.symmaction3.1stordlinear.adjsymm}
\end{align}
For both translations and scalings, 
the second and third symmetry actions here are considered only for
first-order linear adjoint-symmetries \eqref{firstord.linear.adjsymm}--\eqref{firstord.linear.R_Q}. 
\end{corollary}

The first symmetry actions \eqref{translation.symmaction2.adjsymm} and \eqref{scaling.symmaction2.adjsymm}
are a generalization of the same actions derived on multipliers 
in \Ref{Anc2017,Anc2016}. 
The other results are new.

\section{Generalized pre-symplectic and pre-Hamiltonian structures\\ (Noether operators) from symmetry actions}\label{sec:noetherops}

It will be useful to begin with a general discussion. 
Let
\begin{align}
\symmsp_G:= & \{ P^\alpha(x,u^{(k)}), k\geq 0,\text{ s.t. } G'(P)^A|_\Esp=0\}
\label{symmsp}
\\
\adjsymmsp_G:= & \{ Q_A(x,u^{(k)}), k\geq 0,\text{ s.t. } G'{}^*(Q)_\alpha|_\Esp=0\}
\label{adjsymmsp}
\end{align}
denote the linear spaces of symmetries and adjoint-symmetries
for a given PDE system $G^A(x,u^{(N)})=0$. 
Also, let 
\begin{equation}
\multrsp_G:= \{ \Lambda_A(x,u^{(k)}), k\geq 0,\text{ s.t. } G'{}^*(\Lambda)_\alpha + \Lambda'{}^*(G)_\alpha =0 \}
\end{equation}
denote the linear space of multipliers,
which is a subspace of the linear space of adjoint-symmetries \eqref{adjsymmsp}.

Suppose that the PDE system possesses the extra structure
\begin{equation}\label{J.op}
\Dop G' = G'{}^* \Jop 
\end{equation}
where $\Dop$ and $\Jop$ are linear differential operators in total derivatives 
whose coefficients are non-singular on $\Esp$. 
Then, for any symmetry $P^\alpha$, 
$G'{}^*(\Jop(P))|_\Esp=\Dop G'(P)|_\Esp =0$
shows that 
\begin{equation}\label{J.PtoQ}
Q_A:= \Jop(P)_A
\end{equation}
is an adjoint-symmetry. 
If $\Jop(P)_A$ is a multiplier, 
then $\Jop$ represents a \emph{pre-symplectic operator} for the PDE system,
in the sense that it is a mapping from $\symmsp_G$ into $\multrsp_G$,
analogous to a symplectic operator in the case of Hamiltonian systems. 
When $\Jop(P)_A$ is an adjoint-symmetry but not a multiplier, 
it will be called a \emph{Noether operator} \cite{FucFok}. 

Similarly, 
suppose that a PDE system \eqref{pde.sys} possesses the extra structure
\begin{equation}\label{H.op}
\Dop G'{}^*  = G'\Hop
\end{equation}
where $\Dop$ and $\Hop$ are linear differential operators in total derivatives 
whose coefficients are non-singular on $\Esp$. 
For any adjoint-symmetry $Q_A$, 
$G'(\Hop(Q))|_\Esp=\Dop G'{}^*(Q)|_\Esp =0$
whereby 
\begin{equation}
P^\alpha :=\Hop(Q)^\alpha
\end{equation}
is a symmetry. 
Since $\Hop$ is a mapping from $\adjsymmsp_G\supseteq\multrsp_G$ into $\symmsp_G$,
it represents a \emph{pre-Hamiltonian operator} (or \emph{inverse Noether operator})  
for the PDE system, 
analogous to a Hamiltonian operator in the case of Hamiltonian systems \cite{FucFok}.

When the inverses of $\Jop$ and $\Hop$ are well defined, 
then $\Jop^{-1} :=\Hop$ defines a pre-Hamiltonian (inverse Noether) operator, 
and $\Hop^{-1} :=\Jop$ defines a Noether operator. 

These definitions can be generalized to allow 
$\Jop$, $\Hop$, and $\Dop$ to be linear operators 
in partial derivatives with respect to jet space variables in addition to total derivatives. 
In this case, $\Jop$ and $\Hop$ will be respectively called 
a \emph{generalized pre-symplectic (Noether) structure} 
and a \emph{generalized pre-Hamiltonian (inverse Noether) structure}. 

\begin{remark}\label{rem:symplectic.2form}
For $\Hop$ to be a Hamiltonian structure, 
there must exist a non-degenerate integral pairing $\PQpair{Q}{P}$ (modulo total derivatives) 
between symmetries and adjoint-symmetries 
such that $\{Q_1,Q_2\}_\Hop := \PQpair{Q_1}{\Hop(Q_2)}$ 
is a Poisson bracket, 
namely it must be skew-symmetric and satisfy the Jacobi identity. 
Similarly, for $\Jop$ to be a symplectic structure, 
the analogous bilinear-form $\w_\Jop(P_1,P_2) := \PQpair{\Jop(P_1)}{P_2}$
must be skew-symmetric and closed. 
\end{remark}

Now, it will be shown how an action of symmetries on adjoint-symmetries
can be used itself to define a generalized pre-symplectic (Noether) structure 
and, when its inverse exists, a generalized pre-Hamiltonian (inverse Noether) structure.

Consider, in general, any symmetry action
\begin{equation}\label{S_P.op}
Q_A\overset{{\X_P}}{\longrightarrow} S_P(Q)_A
\end{equation}
on $\adjsymmsp_G$, 
where $S_P$ is a linear operator which is also linear in $P^\alpha$. 
Note that $S_P$ may be constructed from both
total derivatives $D_I$ and partial derivatives $\partial_{u^\alpha_I}$.
The action $S_P(Q)_A$ also defines a dual linear operator
\begin{equation}\label{S_Q.op}
S_Q(P)_A:= S_P(Q)_A
\end{equation}  
from $\symmsp_G$ into $\adjsymmsp_G$,
which constitutes a generalized pre-symplectic (Noether) structure. 
For a fixed adjoint-symmetry $Q_A$, 
$S_Q$ will have an inverse $S_Q^{-1}$
which is defined modulo its kernel, 
$\ker(S_Q)\subset\symmsp_G$, 
and which acts on the linear subspace given by its range, 
$S_Q(\symmsp_G) \subseteq\adjsymmsp_G$.
This inverse $S_Q^{-1}$ constitutes a generalized pre-Hamiltonian (inverse Noether) structure 
when $S_Q(\symmsp_G) = \adjsymmsp_G$, 
and otherwise it is a restricted type of that structure. 

From the three symmetry actions in Theorem~\ref{thm:symmactions.adjsymm},
the following structures are obtained. 

\begin{theorem}\label{thm:symmaction.structures}
For a general PDE system \eqref{pde.sys}, 
let $Q_A$ be any fixed adjoint-symmetry. 
Then, 
a generalized Noether structure is given by 
the first symmetry action \eqref{symmaction2.adjsymm}, 
\begin{equation}\label{Jop.symmaction2}
\Jop_1(P)_A := S_{1\,Q}(P)_A = Q'(P)_A + R_{P}^*(Q)_A ;
\end{equation}
a generalized pre-symplectic structure is given by 
and the second symmetry action \eqref{symmaction1.adjsymm}, 
\begin{equation}\label{Jop.symmaction1}
\Jop_2(P)_A := S_{2\,Q}(P)_A = R_P^*(Q)_A - R_Q^*(P)_A ;
\end{equation}
a Noether operator is given by the third symmetry action \eqref{symmaction3.adjsymm}, 
\begin{equation}\label{Jop.symmaction3}
\Jop_Q := S_{3\,Q} = Q' + R_Q^* . 
\end{equation}
The formal inverse of each structure \eqref{Jop.symmaction2} and \eqref{Jop.symmaction1}
gives a generalized pre-Hamiltonian (inverse Noether) structure,
while the formal inverse of the operator \eqref{Jop.symmaction3} 
gives a pre-Hamiltonian (inverse Noether) operator. 
\end{theorem}

The statement about the inverse of $\Jop_Q$ is proven as follows, 
relying on a direct derivation of the symmetry action 
$S_{3\,P}(Q)_A = Q'(P)_A + R_Q^*(P)_A$. 
Similar proofs hold for the inverse of $\Jop_1$ and $\Jop_2$, 
using the derivations that were given in establishing Theorem~\ref{thm:symmactions.adjsymm}. 

For any set of differential functions $P^\alpha$, one has
$\pr\X_P(G'{}^*(Q)_\alpha - R_Q(G)_\alpha)=0$
from the determining equation \eqref{adjsymm.deteqn.offsoln},
where $Q_A$ is any fixed adjoint-symmetry. 
One also has 
$E_{u^\alpha}(P^\beta G'{}^*(Q)_\beta - Q_A G'(P)^A) = 0$
from the adjoint relation \eqref{symm.adjsymm.ibp}. 
These two expressions can be simplified, on $\Esp$, by the following steps
with $H^A:=G'(P)^A - R_P(G)^A$: 
\begin{equation}\label{euler.term}
\begin{aligned}
E_{u^\alpha}(P^\beta G'{}^*(Q)_\beta - Q_A G'(P)^A)|_\Esp
& = E_{u^\alpha}( P^\beta R_Q(G) _\beta -Q_A R_P(G)^A )|_\Esp  -E_{u^\alpha}(Q_A H^A)|_\Esp
\\
& = E_{u^\alpha}(G^A(R_Q^*(P) -R_P^*(Q))_A)|_\Esp  -E_{u^\alpha}(Q_A H^A)|_\Esp
\\
& = G'{}^*(R_Q^*(P) -R_P^*(Q))_\alpha|_\Esp  
-Q'{}^*(H)_\alpha|_\Esp -H'{}^*(Q)_\alpha|_\Esp  ,
\end{aligned}
\end{equation}
which has used the product rule for the Euler operator and integration by parts;
and 
\begin{equation}\label{XP.term}
\begin{aligned}
(\pr\X_P(G'{}^*(Q) - R_Q(G))_\alpha)|_\Esp
& = G'{}^*(Q'(P))_\alpha|_\Esp + (G'(P))'{}^*(Q)_\alpha|_\Esp 
- R_Q(G'(P))_\alpha)|_\Esp 
\\
& (\pr\X_F f'{}^*) = (\pr\X_F f)'{}^* - F'{}^* f'{}^* . 
\\
& = G'{}^*(Q'(P))_\alpha|_\Esp + (R_P(G))'{}^*(Q)_\alpha|_\Esp 
+H'{}^*(Q)_\alpha|_\Esp - R_Q(H)_\alpha)|_\Esp 
\\
& = G'{}^*(Q'(P) +R_P^*(Q))_\alpha|_\Esp 
+H'{}^*(Q)_\alpha|_\Esp - R_Q(H)_\alpha)|_\Esp ,
\end{aligned}
\end{equation}
which has used the identity \eqref{X.adjFrechet.id}, 
combined with the adjoint-symmetry determining equation \eqref{adjsymm.deteqn},
in addition to $(R_P(G))'{}^*|_\Esp = (R_P G')^*|_\Esp = G'{}^* R_P^*|_\Esp$. 
Then, combining the two expressions \eqref{euler.term} and \eqref{XP.term}, 
both of which vanish, 
one obtains
\begin{equation}
\begin{aligned}
0 & =
E_{u^\alpha}(P^\beta G'{}^*(Q)_\beta - Q_A G'(P)^A)|_\Esp
+ (\pr\X_P(G'{}^*(Q) - R_Q(G))_\alpha)|_\Esp
\\
& = G'{}^*(Q'(P) +R_Q^*(P))_\alpha|_\Esp 
-(Q'{}^*(H) + R_Q(H))_\alpha)_\Esp .
\end{aligned}
\end{equation}
This equation shows that 
$G'{}^*(S_{3\,Q}(P))|_\Esp =0$ iff $(Q'{}^*+R_Q)H_\Esp=0$. 
Assuming that $Q'{}^*+R_Q$ is formally invertible, 
one can conclude that $H^A|_\Esp =0$ whenever $S_{3\,Q}(P)_A$ is an adjoint-symmetry,
thereby showing that $P^\alpha$ is a symmetry. 
Hence, $S_{3\,Q}^{-1}=\Jop_Q^{-1}$ maps adjoint-symmetries into symmetries. 
This completes the proof. 

It is worth noting that this proof gives the relation 
$G'{}^*(\Jop_Q(P)) = \Jop_Q^*(G'(P))$ on $\Esp$,
which is stronger than the structure \eqref{J.op}. 

Finally, 
the Noether operator \eqref{Jop.symmaction3} can be combined with 
the Frechet derivative identity \eqref{symm.adjsymm.ibp} to construct a bilinear form 
as follows. 

\begin{proposition}\label{prop:noether.pairing}
Let $Q_A$ be any fixed adjoint-symmetry such that the Noether operator \eqref{Jop.symmaction3} is non-trivial,
and let $\Psi^i(P,Q)$ be the components of the vector 
in the Frechet derivative identity \eqref{symm.adjsymm.ibp}. 
A bilinear form on the linear space of symmetries $P^\alpha\partial_{u^\alpha}$ is 
defined by 
\begin{equation}\label{noether.symm.pairing}
\w_Q(P_1,P_2) = \int_{\Omega} \Psi^i(P_1,J_Q(P_2)) \hat n_i\,d^{n-1}V
\end{equation}
where $\Omega$ is a domain of codimension $1$ in $\Rnum^n$,
with $\hat n_i$ denoting a unit normal 1-form of $\Omega$,
and with $d^{n-1}V$ denoting the volume element on $\Omega$. 
\end{proposition}

Further developments related to the structures 
in Theorem~\ref{thm:symmaction.structures} and Proposition~\ref{prop:noether.pairing}
will given for evolution PDEs in section~\ref{sec:evolPDEs}. 

\underline{Running example}: 
The p-gKdV equation \eqref{p-gkdv} has 
$\symmsp_\text{p-gKdV}= \spn(P_1,P_2,P_3,P_4)$ for Lie point symmetries 
and $\adjsymmsp_\text{p-gKdV} =\spn(Q_1,Q_2,Q_3)$ for low-order adjoint-symmetries. 
The dual linear maps given by the three symmetry actions \eqref{symmaction2.adjsymm}, \eqref{symmaction1.adjsymm}, \eqref{symmaction3.adjsymm} 
using $Q=c_1Q_1+c_2Q_2+c_3Q_3$ 
are shown in Table~\ref{table:p-gkdv.dualsymmactions},
where $c_1,c_2,c_3$ are arbitrary constants. 
The corresponding structures \eqref{Jop.symmaction2}, \eqref{Jop.symmaction1}, \eqref{Jop.symmaction3} 
are defined by 
$\Jop_i(a_1P_1 +a_2P_2 +a_3P_3+a_4P_4)= a_1S_{i\,Q}(P_1) + a_2S_{i\,Q}(P_2)+a_3S_{i\,Q}(P_3)+a_4S_{Q}(P_4)$,
$i=1,2,3$,
where $a_1,a_2,a_3,a_4$ are arbitrary constants. 
In particular, the explicit form of the Noether operator is 
\begin{equation}\label{p-gkdv.noetherop}
\Jop_Q = Q_3' + R_{Q_3}^* = 2(2-p)D_x ,
\end{equation}
since 
$Q_3'=(2u_x + 3p t u_{t,x} + p x u_{xx})'
=2D_x + 3p t D_tD_x + p x D_x^2$
and 
$R_{Q_3}^* = (-p x D_x^2 -3p t D_tD_x - (3p+2)D_x)^* 
= 2(1-p)D_x -p x D_x^2 -3p t D_tD_x$
from equations \eqref{p-gkdv.Q} and \eqref{p-gkdv.RQ}. 

\begin{table}[h!]
\caption{dual of p-gKdV symmetry actions with $Q=c_1Q_1+c_2Q_2+c_3Q_3$} 
\label{table:p-gkdv.dualsymmactions}
\begin{tabular}{l||c|c|c}
& $S_{1\,Q}$
& $S_{2\,Q}$
& $S_{3\,Q}$
\\
\hline 
\hline
$P_1$
& $0$
& $0$
& $0$
\\
\hline
$P_2$
& $p c_3 Q_1$
& $(4-p) c_3 Q_1$
& $2(p-2) c_3 Q_1$
\\
\hline
$P_3$
& $3p c_3 Q_2$
& $(p+4) c_3 Q_2$
& $2(p-2) c_3 Q_2$
\\
\hline
$P_4$
& $(p-4)c_1 Q_1 -(p+4)c_2 Q_2 + 2(p-2)c_3 Q_3$
& $(p-4)c_1 Q_1 -(p+4)c_2 Q_2$
& $2(p-2)c_3 Q_3$
\\
\end{tabular}
\end{table}

\section{Bracket structures for adjoint-symmetries}\label{sec:adjsymmbracket}

The commutator \eqref{symm.commutator} of symmetries
defines a Lie bracket on the linear space of symmetries \eqref{symmsp}. 
An interesting fundamental question is whether there exists
any bilinear bracket on the linear space of adjoint-symmetries \eqref{adjsymmsp}. 
Such a structure would allow the possibility for
a pair of known adjoint-symmetries to generate a new adjoint-symmetry, 
just as a pair of known symmetries can generate a new symmetry. 
Additionally, if a bilinear bracket has a projection into the linear space of multipliers, 
then this would provide a generalization of a Poisson bracket. 

Every action of symmetries on adjoint-symmetries will now be shown to give rise to 
two different bilinear bracket structures on adjoint-symmetries.
The first bracket is a Lie bracket constructed 
from the pull-back of the symmetry commutator \eqref{symm.commutator} 
under an inverse of the symmetry action on adjoint-symmetries. 
This yields a homomorphism from the Lie algebra of symmetries into a Lie algebra of adjoint-symmetries. 
The second bracket does not involve the symmetry commutator \eqref{symm.commutator} 
and instead uses the symmetry action composed with an inverse action to construct a recursion operator on adjoint-symmetries. 

These constructions will be carried out in terms of the dual linear operator \eqref{S_Q.op}
associated to a general symmetry action \eqref{S_P.op}. 
Afterward, the properties of the resulting brackets will be discussed
for each of the three actions \eqref{symmaction2.adjsymm}, 
\eqref{symmaction3.adjsymm}.

\subsection{Adjoint-symmetry commutator brackets from symmetry actions}

The construction of the first bracket goes as follows. 

\begin{proposition}\label{prop:adjsymm.bracket}
Fix an adjoint-symmetry $Q_A$ in $\adjsymmsp_G$,
and let $S_Q$ be the dual linear operator \eqref{S_Q.op}
associated to a symmetry action $S_P$ on $\adjsymmsp_G$.
If the kernel of $S_Q$ is an ideal in $\symmsp_G$,
then 
\begin{equation}\label{adjsymm.bracket}
{}^Q[Q_1,Q_2]_A := S_Q([S_Q^{-1}Q_1,S_Q^{-1}Q_2])_A
\end{equation}
defines a bilinear bracket on the linear space $S_Q(\symmsp_G)\subseteq\adjsymmsp_G$. 
This bracket can be expressed as
\begin{equation}\label{adjsymm.bracket.explicit}
{}^Q[Q_1,Q_2]_A 
 = Q_2'(S_Q^{-1}Q_1) - Q_1'(S_Q^{-1}Q_2) - S_Q'(S_Q^{-1}Q_2)(S_Q^{-1}Q_1) + S_Q'(S_Q^{-1}Q_1)(S_Q^{-1}Q_2)
\end{equation}
where $S_Q'$ denotes the Frechet derivative of $S_Q$. 
\end{proposition}

Any one of the symmetry actions \eqref{symmaction2.adjsymm}, \eqref{symmaction1.adjsymm}, \eqref{symmaction3.adjsymm}
can be used to write down formally a corresponding bracket \eqref{adjsymm.bracket}.
However, $S_Q^{-1}$ is well-defined only modulo $\ker(S_Q)$,
and so in the absence of any extra structure to fix this arbitrariness, 
the condition that $\ker(S_Q)$ is an ideal is necessary and sufficient 
for the bracket to be well defined 
(namely, invariant under $S_Q^{-1}\rightarrow S_Q^{-1}+\ker(S_Q)$). 
This condition will select a set of adjoint-symmetries $Q_A$ 
that can be used in constructing the bracket. 
When $\ker(S_Q)$ is an ideal, so is $\ker(S_{\lambda Q}) = \lambda \ker(S_{Q})$, 
for any constant $\lambda$. 
Hence, the set of adjoint-symmetries $Q_A$ for which $\ker(S_Q)$ is an ideal
comprises a projective subspace in $\adjsymmsp_G$. 
In the case when the dimension of this subspace is larger than $1$, 
it is natural to select $Q_A$ such that $\ran(S_Q)$ is maximal in $\adjsymmsp_G$. 

For the linear space $\ker(S_Q)\subseteq\symmsp_G$ to be an ideal, 
it must be a subalgebra that is preserved by the action of $\symmsp_G$ 
given by the Lie bracket \eqref{symm.commutator}. 
The subalgebra condition 
\begin{equation}\label{kerSQ.condition}
[\ker(S_Q),\ker(S_Q)] \subseteq \ker(S_Q)
\end{equation}
states that $S_Q([P_1,P_2])=0$ is required to hold 
for all pairs of symmetries
$\X_{P_1}=P_1^\alpha\partial_{u^\alpha}$ and $\X_{P_2}=P_2^\alpha\partial_{u^\alpha}$
such that 
$S_Q(P_1)_A= S_{P_1}(Q)_A = 0$ and $S_Q(P_2)_A= S_{P_2}(Q)_A = 0$.
The question of whether this condition \eqref{kerSQ.condition}
is satisfied by each of the three symmetry actions 
will now be addressed. 

For the first symmetry action \eqref{symmaction2.adjsymm}, 
consider 
\begin{equation}\label{2ndaction.ker.eqns}
0= S_{1\,Q}(P_1)_A = Q'(P_1)_A + R_{P_1}^*(Q)_A ,
\quad
0= S_{1\,Q}(P_2)_A = Q'(P_2)_A + R_{P_2}^*(Q)_A . 
\end{equation}
Applying the symmetries $\X_{P_2}$ and $\X_{P_1}$
respectively to these two equations and subtracting them 
yields
\begin{equation}\label{2ndaction.commutator.eqn}
\begin{aligned}
0 = & \pr\X_{P_2}(Q'(P_1)_A + R_{P_1}^*(Q)_A) - \pr\X_{P_1}(Q'(P_2)_A + R_{P_2}^*(Q)_A) \\
= & Q'(\pr\X_{P_2}(P_1) -\pr\X_{P_1}(P_2))_A 
+ \pr\X_{P_2}(R_{P_1}^*)(Q)_A - \pr\X_{P_1}(R_{P_2}^*)(Q)_A 
\\&\qquad
+ R_{P_1}^*(\pr\X_{P_2}(Q))_A - R_{P_2}^*(\pr\X_{P_1}(Q))_A
\end{aligned}
\end{equation}
using $\pr\X_{P_2}Q'(P_1) - \pr\X_{P_1}Q'(P_2) = Q''(P_2,P_1) - Q''(P_1,P_2) = 0$. 
The first term in equation \eqref{2ndaction.commutator.eqn} 
reduces to a commutator expression 
\begin{equation}\label{2ndaction.commutator.eqn.term1}
Q'(\pr\X_{P_2}(P_1) -\pr\X_{P_1}(P_2))_A = Q'([P_2,P_1])_A . 
\end{equation}
The middle two terms can be expressed as
\begin{equation}\label{2ndaction.commutator.eqn.term2}
\pr\X_{P_2}(R_{P_1}^*)(Q)_A -\pr\X_{P_1}(R_{P_2}^*)(Q)_A 
= R_{[P_2,P_1]}^*(Q)_A - [R_{P_2}^*,R_{P_1}^*](Q)_A 
\end{equation}
by use of the identity
\begin{equation}
R_{[P_2,P_1]}^* = [R_{P_2}^*,R_{P_1}^*] + \pr\X_{P_2}(R_{P_1}^*) -\pr\X_{P_1}(R_{P_2}^*)
\end{equation}
which can derived straightforwardly from the symmetry determining equation \eqref{symm.deteqn.offsoln} off of $\Esp$. 
Next, the last term in equation \eqref{2ndaction.commutator.eqn.term2} 
can be combined with
the last two terms in equation \eqref{2ndaction.commutator.eqn}, 
yielding
\begin{equation}
R_{P_1}^*(Q'(P_2) +R_{P_2}^*(Q))_A - R_{P_2}^*(Q'(P_1) +R_{P_1}^*(Q))_A 
=0
\end{equation}
due to equations \eqref{2ndaction.ker.eqns}. 
Hence, after these simplifications, 
equation \eqref{2ndaction.commutator.eqn} becomes 
$0=Q'([P_2,P_1])_A + R_{[P_2,P_1]}^*(Q)_A = S_{1\,Q}([P_2,P_1])_A$. 
This establishes the following result. 

\begin{lemma}\label{lem:symmaction2.ker.subalgebra}
For the first symmetry action \eqref{symmaction2.adjsymm},
$\ker(S_Q)$ is a subalgebra in $\symmsp_G$. 
\end{lemma}

To continue,
consider the third symmetry action \eqref{symmaction3.adjsymm}. 
Similar steps will now be carried out, starting from
\begin{equation}\label{3rdaction.ker.eqns}
0= S_{3\,Q}(P_1)_A = Q'(P_1)_A + R_Q^*(P_1)_A ,
\quad
0= S_{3\,Q}(P_2)_A = Q'(P_2)_A + R_Q^*(P_2)_A . 
\end{equation}
Respectively applying the symmetries $\X_{P_2}$ and $\X_{P_1}$
to these two equations and subtracting, 
one obtains
\begin{equation}\label{3rdaction.commutator.eqn}
0 = Q'([P_2,P_1])_A +R_Q^*([P_2,P_1])_A 
+ \pr\X_{P_2}(R_Q^*)(P_1)_A -\pr\X_{P_1}(R_Q^*)(P_2)_A . 
\end{equation}
Hence, one sees that 
$S_{3\,Q}([P_2,P_1]) = Q'([P_2,P_1])_A +R_Q^*([P_2,P_1])_A 
= \pr\X_{P_1}(R_Q^*)(P_2)_A -\pr\X_{P_2}(R_Q^*)(P_1)_A$
does not vanish in general.
This represents an obstruction to the bracket being well-defined. 
A useful remark is that if $Q=\Lambda$ is a conservation law multiplier 
for a PDE system with no differential identities (Lemma~\ref{lem:no.diff.identities}),
then the relation \eqref{multr.relation}
shows that 
\begin{equation}
\begin{aligned}
\pr\X_{P_1}(R_Q^*)(P_2)_A -\pr\X_{P_2}(R_Q^*)(P_1)_A
& = \pr\X_{P_2}(Q')(P_1)_A -\pr\X_{P_1}(Q')(P_2)_A \\
& = Q''(P_1,P_2) - Q''(P_2,P_1) =0
\end{aligned}
\end{equation}
whereby the obstruction vanishes. 

A similar obstruction arises for the bracket given by the second symmetry action \eqref{symmaction1.adjsymm}.
Specifically, by the same steps used for the first and third symmetry actions, 
one obtains 
$S_{2\,Q}([P_2,P_1]) = R_{[P_2,P_1]}^*(Q)_A -R_Q^*([P_2,P_1])_A 
= \pr\X_{P_2}(R_Q^*)(P_1)_A -\pr\X_{P_1}(R_Q^*)(P_2)_A
+R_{P_2}^*(S_{1\,Q}(P_1))_A - R_{P_1}^*(S_{1\,Q}(P_2))_A$. 
This expression contains the same obstruction terms as for the third symmetry action, 
as well as terms that involve the first symmetry action itself. 
If $Q=\Lambda$ is a conservation law multiplier for a PDE system 
with no differential identities (Lemma~\ref{lem:no.diff.identities}),
then this obstruction vanishes. 

Consequently, the following two results have been established. 

\begin{lemma}\label{lem:symmaction3.ker.subalgebra}
For the third symmetry action \eqref{symmaction3.adjsymm}, 
$\ker(S_{3\,Q})$ is a subalgebra in $\symmsp_G$
iff the condition  
\begin{equation}\label{symmaction3.subalg.condition}
\pr\X_{P_2}(R_Q^*)(P_1)_A - \pr\X_{P_1}(R_Q^*)(P_2)_A =0
\end{equation}
holds for all symmetries
$\X_{P_1}=P_1^\alpha\partial_{u^\alpha}$ and $\X_{P_2}=P_2^\alpha\partial_{u^\alpha}$
in $\ker(S_{3\,Q})$. 
\end{lemma}

\begin{lemma}\label{lem:symmaction1.ker.subalgebra}
For the second symmetry action \eqref{symmaction1.adjsymm}, 
$\ker(S_{2\,Q})$ is a subalgebra in $\symmsp_G$
iff the condition  
\begin{equation}\label{symmaction1.subalg.condition}
\pr\X_{P_2}(R_Q^*)(P_1)_A -\pr\X_{P_1}(R_Q^*)(P_2)_A
+R_{P_2}^*(S_{1\,Q}(P_1))_A - R_{P_1}^*(S_{1\,Q}(P_2))_A
=0
\end{equation}
holds for all symmetries
$\X_{P_1}=P_1^\alpha\partial_{u^\alpha}$ and $\X_{P_2}=P_2^\alpha\partial_{u^\alpha}$
in $\ker(S_{2\,Q})$. 
\end{lemma}

The preceding developments can be summarized as follows. 

\begin{proposition}\label{prop:bracket.condition}
The adjoint-symmetry commutator bracket \eqref{adjsymm.bracket}
associated to each of the symmetry actions \eqref{symmaction2.adjsymm}, \eqref{symmaction1.adjsymm}, \eqref{symmaction3.adjsymm}
is well-defined on $S_Q(\symmsp_G)\subseteq \adjsymmsp_G$
if $\ad(\symmsp_G)\ker(S_Q) \subseteq \ker(S_Q)$ 
and, for the actions \eqref{symmaction1.adjsymm} and \eqref{symmaction3.adjsymm}, 
if the respective conditions \eqref{symmaction1.subalg.condition} and \eqref{symmaction3.subalg.condition} hold when $\dim\ker(S_Q)>1$. 
These latter conditions are identically satisfied
when $Q$ is a conservation law multiplier for a PDE system 
with no differential identities. 
\end{proposition}

An alternative way to have the bracket be well defined is 
if the quotient $\symmsp_G/\ker(S_Q)$ can be naturally identified with 
a subspace in $\symmsp_G$. 
This is equivalent to requiring that the symmetry Lie algebra 
admits an extra structure of a direct-sum decomposition as a linear space
\begin{equation}\label{S_Q:decomp}
\symmsp_G = \ker(S_Q)\oplus\coker(S_Q)
\end{equation}
such that the decomposition is independent of a choice of basis. 
Then $S_Q^{-1}$ can be defined as belonging to the subspace $\coker(S_Q)$, 
and hence the bracket will be well defined. 

It will now be shown that the extra structure \eqref{S_Q:decomp} will typically exist 
for a symmetry Lie algebra that contains a scaling symmetry \eqref{scaling.symm}. 

Every symmetry in $\symmsp_G$ can be decomposed 
into a sum of symmetries that are scaling homogeneous. 
Consequently, there will exist a basis for $\symmsp_G$ consisting of 
$P_\scal$ and $\{P_k\}_{k=1,\ldots,\dim(\symmsp_G)-1}$, 
such that $[P_\scal,P_k] = r^{(k)} P_k$
where the constant $r_{(k)}$ is the scaling weight of the symmetry $P_k$. 
Then there exists a direct-sum decomposition 
\begin{equation}\label{scaling.decomp}
\symmsp_G = \spn(P_\scal) \oplus\sum_k{}^{\textstyle\oplus}\spn(P_k) . 
\end{equation}
which is basis independent. 
This will provide the extra structure \eqref{S_Q:decomp} if 
the subspaces $\ker(S_Q)$ and $\coker(S_Q)$ can be uniquely characterized 
in terms of their scaling weights. 

\begin{proposition}\label{prop:bracket.condition.alt}
Suppose $\symmsp_G$ contains a scaling symmetry \eqref{scaling.symm}. 
For each of the symmetry actions \eqref{symmaction2.adjsymm}, \eqref{symmaction1.adjsymm}, \eqref{symmaction3.adjsymm}, 
if $\ker(S_Q)$ is a scaling-homogeneous subspace in $\symmsp_G$, 
then the adjoint-symmetry commutator bracket \eqref{adjsymm.bracket}
is well-defined on the linear space $S_Q(\symmsp_G)\subseteq \adjsymmsp_G$
by taking $S_Q^{-1}$ to belong to a sum of scaling-homogeneous subspaces
with scaling weights that are different than that of $\ker(S_Q)$. 
\end{proposition}

This result can be generalized 
if $\ker(S_Q)$ is a direct sum of scaling homogeneous subspaces that have no scaling weights in common with any scaling homogeneous subspace in $\coker(S_Q)$. 

Now, the basic properties of the general adjoint-symmetry commutator bracket \eqref{adjsymm.bracket}
will be studied. 
Recall that the underlying symmetry commutator bracket
is antisymmetric and obeys the Jacobi identity. 
This implies that the same properties are inherited by the bracket \eqref{adjsymm.bracket}. 

\begin{theorem}\label{thm:adjsymm.bracket.properties}
The adjoint-symmetry commutator bracket \eqref{adjsymm.bracket} is a Lie bracket, 
namely it is antisymmetric
\begin{equation}\label{adjsymm.bracket.skew}
{}^Q[Q_1,Q_2]_A +{}^Q[Q_2,Q_1]_A=0
\end{equation}  
and obeys the Jacobi identity
\begin{equation}\label{adjsymm.bracket.jacobi}
{}^Q[Q_1,{}^Q[Q_2,Q_3]]_A + {}^Q[Q_2,{}^Q[Q_3,Q_1]]_A + {}^Q[Q_3,{}^Q[Q_1,Q_2]]_A =0 . 
\end{equation}    
Hence, the linear subspace $S_Q(\symmsp_G)\subseteq\adjsymmsp_G$ of adjoint-symmetries 
acquires a Lie algebra structure which is homomorphic to the symmetry Lie algebra. 
If there exists an adjoint-symmetry $Q_A$ such that $S_Q(\symmsp_G)=\adjsymmsp_G$
where $\ker(S_Q)$ satisfies the conditions in either of 
Propositions~\ref{prop:bracket.condition} and~\ref{prop:bracket.condition.alt}, 
then the whole space $\adjsymmsp_G$ will be a Lie algebra. 
\end{theorem}

Since $S_Q$ is a linear mapping,
the condition $S_Q(\symmsp_G)=\adjsymmsp_G$ can be expressed equivalently as 
\begin{equation}\label{SQ.maximal.condition}
\dim\adjsymmsp_G +\dim\ker(S_Q) = \dim\symmsp_G . 
\end{equation}
Hence, $\dim\symmsp_G \geq \dim\adjsymmsp_G$ is a necessary condition. 
This version is most useful when the dimensions are finite.

\underline{Running example}: 
For the p-gKdV equation \eqref{p-gkdv}, 
the three dual linear maps $S_Q$ with $Q=c_1Q_1+c_2Q_2+c_3Q_3$ 
have a 4-dimensional domain $\spn(P_1,P_2,P_3,P_4)$,
while their range is at most 3-dimensional, since it belongs to $\spn(Q_1,Q_2,Q_3)$.
Hence, the kernel of each $S_Q$ is at least 1-dimensional. 
The specific ranges and kernels, 
which depend on the choice of the constants $c_1,c_2,c_3$, 
can be found from Table~\ref{table:p-gkdv.dualsymmactions} as follows. 

Firstly, for $S_{3\,Q}$, 
there is no loss of generality in taking $c_3=1$, $c_1=c_2=0$. 
The range of $S_{3\,Q}$ is $\spn(Q_1,Q_2,Q_3)$,
while the kernel is $\spn(P_1)$. 
This subspace in $\symmsp_\text{p-gKdV}$ is an ideal,
as shown by the symmetry commutators \eqref{p-gkdv.symmalg}.
The resulting adjoint-symmetry bracket ${}^{Q_3}[\ \cdot\ ,\ \cdot\ ]$ 
is thereby well-defined. 
It is computed by: 
$S_{3\,Q_3}^{-1}(Q_1) = \tfrac{1}{2(p-2)}P_2$,
$S_{3\,Q_3}^{-1}(Q_2) = \tfrac{1}{2(p-2)}P_3$, 
$S_{3\,Q_3}^{-1}(Q_3) = \tfrac{1}{2(p-2)}P_4$,
modulo $\spn(P_1)$,
and thus 
\begin{subequations}
\begin{align}
& {}^{Q_3}[Q_1,Q_2] = S_{3\,Q_3}([\tfrac{1}{2(p-2)}P_2,\tfrac{1}{2(p-2)}P_3])
=\tfrac{1}{4(p-2)^2}S_{3\,Q_3}(0) 
=0, 
\\
& 
{}^{Q_3}[Q_1,Q_3] = S_{3\,Q_3}([[\tfrac{1}{2(p-2)}P_2,\tfrac{1}{2(p-2)}P_4])
= \tfrac{1}{4(p-2)^2} S_{3\,Q_3}(pP_2) 
= \tfrac{p}{2(p-2)}Q_1,
\\
& 
{}^{Q_3}[Q_2,Q_3] = S_{3\,Q_3}([\tfrac{1}{2(2-p)}P_3,\tfrac{1}{2(p-2)}P_4])
= \tfrac{1}{4(p-2)^2} S_{3\,Q_3}(3pP_3) 
= \tfrac{3p}{2(p-2)}Q_2,
\end{align}
\end{subequations}
using the symmetry commutators \eqref{p-gkdv.symmalg}.
This bracket is a Lie bracket such that 
$\spn(Q_1,Q_2,Q_3)$ is homomorphic to the symmetry algebra $\spn(P_1,P_2,P_3,P_4)$
and isomorphic to the subalgebra $\spn(P_2,P_3,P_4)$. 
In particular, 
the bracket can be expressed in terms of the scaled Noether operator $D_x$ 
(cf \eqref{p-gkdv.noetherop}):
\begin{equation}\label{p-gkdv.adjsymm.bracket}
{}^{Q_3}[\ \cdot\ ,\ \cdot\ ] = D_x([D_x^{-1}\ \cdot\ ,D_x^{-1}\ \cdot\ ]) .
\end{equation}

Secondly, for $S_{1\,Q}$, 
the range is maximal if $c_3\neq0$ with any values of $c_2,c_3$. 
The simplest choice is $c_1=c_2=0$, $c_3=1$, namely $Q=Q_3$. 
Then the range and the kernel are the same as for $S_{3\,Q}$,
and so the resulting adjoint-symmetry bracket is well-defined. 
It is computed by: 
$S_{1\,Q_3}^{-1}(Q_1) = \tfrac{1}{p}P_2$,
$S_{1\,Q_3}^{-1}(Q_2) = \tfrac{1}{3p}P_3$, 
$S_{1\,Q_3}^{-1}(Q_3) = \tfrac{1}{2(p-2)}P_4$,
modulo $\spn(P_1)$,
and thus 
\begin{subequations}
\begin{align}
& {}^{Q_3}[Q_1,Q_2] = S_{1\,Q_3}([\tfrac{1}{p}P_2,\tfrac{1}{3p}P_3])
=\tfrac{1}{3p^2}S_{1\,Q_3}(0) 
=0, 
\\
& 
{}^{Q_3}[Q_1,Q_3] = S_{1\,Q_3}([[\tfrac{1}{p}P_2,\tfrac{1}{2(p-2)}P_4])
= \tfrac{1}{2p(p-2)} S_{1\,Q_3}(pP_2) 
= \tfrac{p}{2(p-2)}Q_1,
\\
& 
{}^{Q_3}[Q_2,Q_3] = S_{1\,Q_3}([\tfrac{1}{3p}P_3,\tfrac{1}{2(p-2)}P_4])
= \tfrac{1}{6p(p-2)} S_{1\,Q_3}(3pP_3) 
= \tfrac{3}{2(p-2)}Q_2. 
\end{align}
\end{subequations}
This yields the same bracket as for $S_{3\,Q}$. 

Thirdly, for $S_{2\,Q}$, 
the maximal range is $\spn(Q_1,Q_2)$, 
which arises if $c_3\neq0$ with any values for $c_1,c_2$. 
The kernel is spanned by $P_1$, $P_4 +\tfrac{c_1}{c_3}P_2 +\tfrac{c_2}{c_3}P_3$. 
However, this space is not an ideal in $\symmsp_\text{p-gKdV}$,
as shown by the symmetry commutators \eqref{p-gkdv.symmalg}.
Hence a corresponding adjoint-symmetry bracket cannot be defined 
without the use of extra structure. 
The scaling symmetry $P_4$ is available to provide a direct-sum decomposition 
$\spn(P_1,P_2,P_3,P_4) = \ker(S_{2\,Q})\oplus\coker(S_{2\,Q})$
where, for the choice $c_1=c_2=0$ and $c_3=1$, 
$\ker(S_{2\,Q}) = \spn(P_4)\oplus\spn(P_1)$ 
and $\coker(S_{2\,Q}) = \spn(P_2)\oplus\spn(P_3)$
are characterized by their distinct scaling weights with respect to $\ad(P_4)$: 
$(0,2-p)$ and $(-p,-3p)$. 
Then an adjoint-symmetry bracket can be defined via
$S_{2\,Q_3}^{-1}(Q_1) = \tfrac{1}{4-p} P_2$,
$S_{2\,Q_3}^{-1}(Q_2) = \tfrac{1}{p+4} P_3$, 
in $\coker(S_{2\,Q})$, 
and thus 
\begin{equation}
{}^{Q_3}[Q_1,Q_2] = S_{2\,Q_3}([\tfrac{1}{4-p}P_2,\tfrac{1}{p+4}P_3])
= \tfrac{1}{16-p^2} S_{2\,Q_3}(0) 
= 0. 
\end{equation}
This yields an abelian Lie bracket on the subspace 
$\spn(Q_1,Q_2)\subset \adjsymmsp_\text{p-gKdV}$. 
It coincides with the previous Lie bracket restricted to this subspace. 

These three Lie brackets are summarized in Table~\ref{table:p-gkdv.adjsymm.brackets}. 

\begin{table}[h!]
\caption{p-gKdV equation: adjoint-symmetry Lie brackets}
\label{table:p-gkdv.adjsymm.brackets}
\begin{subtable}{.5\linewidth}
\centering
\caption{bracket using $S_{1\,Q_3}=S_{3\,Q_3}$}
\begin{tabular}{l||c|c|c}
& $Q_1$
& $Q_2$
& $Q_3$
\\
\hline
\hline  
$Q_1$
& $0$
& $0$
& $\tfrac{p}{2(p-2)}Q_1$
\\
\hline
$Q_2$
& 
& $0$
& $\tfrac{3}{2(p-2)}Q_2$
\\
\hline
$Q_3$
& 
& 
& $0$
\\
\end{tabular}
\end{subtable}%
\begin{subtable}{.5\linewidth}
\centering
\caption{bracket using $S_{2\,Q_3}$}
\begin{tabular}{l||c|c}
& $Q_1$
& $Q_2$
\\
\hline
\hline  
$Q_1$
& $0$
& $0$
\\
\hline
$Q_2$
& 
& $0$
\\
\end{tabular}
\end{subtable}%
\end{table}

\subsection{Adjoint-symmetry commutators associated to symmetry subalgebras}

The Lie algebra structure identified in Theorem~\ref{thm:adjsymm.bracket.properties}
motivates a related construction of adjoint-symmetry commutator brackets 
given by a pull-back of Lie subalgebras in $\symmsp_G$ under $S_Q^{-1}$. 

As the starting point, 
the linear subspace $S_Q(\symmsp_G)$ will be replaced by $S_Q({\mathcal A})$
where ${\mathcal A}$ is any Lie subalgebra in $\symmsp_G$
and where $Q_A$ is chosen such that $\ker(S_Q)\cap {\mathcal A}$ is empty. 
The set of such adjoint-symmetries will, as before, 
be a projective subspace in $\adjsymmsp_G$. 

Then the construction of the commutator bracket given in Proposition~\ref{prop:adjsymm.bracket} is modified as follows. 

\begin{proposition}\label{prop:adjsymm.bracket.subalg}
Given a Lie subalgebra $\mathcal A$ in $\symmsp_G$
and a symmetry action $S_P$ on $\adjsymmsp_G$, 
fix an adjoint-symmetry $Q_A$ in $\adjsymmsp_G$ 
such that the kernel of $S_Q$ restricted to $\mathcal A$ is empty,
where $S_Q$ is the dual linear operator \eqref{S_Q.op} of the symmetry action. 
Then the commutator bracket \eqref{adjsymm.bracket} is well-defined 
on the linear space $S_Q(\symmsp_G) \subseteq\adjsymmsp_G$,
and this structure is isomorphic to the Lie subalgebra $\mathcal A$. 
\end{proposition}

In particular, $S_Q^{-1}$ provides an isomorphism under which 
the commutator bracket \eqref{adjsymm.bracket} on $S_Q(\symmsp_G)$ 
is the pull-back of the Lie bracket on $\mathcal A$. 
The condition 
\begin{equation}\label{adjsymm.bracket.subalg.cond}
\ker(S_Q)\cap {\mathcal A}=\emptyset
\end{equation}
will select the adjoint-symmetries $Q_A$ that can be used in constructing this bracket. 
If this condition fails to be satisfied by all adjoint-symmetries, 
then it implies that there is no subspace in $\adjsymmsp_G$ on which the bracket 
produces a Lie algebra isomorphic to $\mathcal A$. 

The question of which Lie subalgebras $\mathcal A$ in $\symmsp_G$ 
have counterparts in $\adjsymmsp_G$ for a PDE system $G^A=0$ 
thereby becomes an interesting algebraic classification problem.

\subsection{Adjoint-symmetry non-commutator brackets from symmetry actions}

The construction of the second bracket disregards the symmetry commutator 
but lacks the attendant properties. 

\begin{proposition}\label{prop:adjsymm.bracket2}
Fix an adjoint-symmetry $Q_A$ in $\adjsymmsp_G$,
and let $S_Q$ be the dual linear operator \eqref{S_Q.op}
associated to a symmetry action $S_P$ on $\adjsymmsp_G$.
If the kernel of $S_Q$ satisfies 
\begin{equation}\label{bracket2.Scondition}
S_P=0 \text{ for all } P\in\ker(S_Q), 
\end{equation}
then a bilinear bracket from $\adjsymmsp_G\times S_Q(\symmsp_G)$ into $S_Q(\symmsp_G)\subseteq\adjsymmsp_G$ 
is defined by 
\begin{equation}\label{adjsymm.bracket2}
{}^Q(Q_1,Q_2)_A := S_{Q_1}(S_Q^{-1}Q_2)_A . 
\end{equation}
\end{proposition}

Any one of the symmetry actions \eqref{symmaction1.adjsymm}, \eqref{symmaction2.adjsymm}, \eqref{symmaction3.adjsymm}
can be used to write down formally a corresponding bracket \eqref{adjsymm.bracket2}.
Note that, 
unlike the situation for the commutator bracket \eqref{adjsymm.bracket}, 
the condition \eqref{bracket2.Scondition} only involves the properties of the symmetry action $S_Q$
and does not depend on the Lie algebra structure of $\symmsp_G$. 
This condition can be by-passed when a scaling symmetry \eqref{scaling.symm} is contained in the symmetry Lie algebra. 

\begin{proposition}\label{prop:adjsymm.bracket2.alt}
Suppose $\symmsp_G$ contains a scaling symmetry \eqref{scaling.symm}. 
For any symmetry action, 
if $\ker(S_Q)$ is a scaling-homogeneous subspace in $\symmsp_G$, 
then the adjoint-symmetry bracket \eqref{adjsymm.bracket2}
is well-defined on $S_Q(\symmsp_G)\subseteq \adjsymmsp_G$
by taking $S_Q^{-1}$ to belong to a sum of scaling-homogeneous subspaces. 
\end{proposition}

In contrast to the commutator bracket \eqref{adjsymm.bracket}, 
the bracket \eqref{adjsymm.bracket2} is non-symmetric. 
Its only general property is that
\begin{equation}\label{adjsymm.bracket2.property}
{}^Q(Q,Q_2)= Q_2
\end{equation}
for all $Q_2$ in the linear subspace $S_Q(\symmsp_G) \subseteq \adjsymmsp_G$.

There are two worthwhile remarks that can be made. 

\begin{remark}
(i) The bracket \eqref{adjsymm.bracket2} can be viewed 
as arising from the property that $S_{Q_1}S_{Q_2}^{-1}$ is a recursion operator 
on adjoint-symmetries in $S_Q(\symmsp_G)$. 
(ii) A symmetric version and a skew-symmetric version of the bracket \eqref{adjsymm.bracket2} 
can be defined by respectively symmetrizing and antisymmetrizing 
on the pair $Q_1$ and $Q_2$: 
\begin{equation}\label{adjsymm.skewbracket}
{}^Q(Q_1,Q_2)^{-}_A := \tfrac{1}{2}\big( S_{Q_1}(S_Q^{-1}Q_2)_A - S_{Q_2}(S_Q^{-1}Q_1)_A \big)
\end{equation}
and
\begin{equation}\label{adjsymm.antiskewbracket}
{}^Q(Q_1,Q_2)^{+}_A := \tfrac{1}{2}\big( S_{Q_1}(S_Q^{-1}Q_2)_A + S_{Q_2}(S_Q^{-1}Q_1)_A \big) . 
\end{equation}
\end{remark}

The recursion operator $S_{Q_1}S_{Q_2}^{-1}$ was derived originally 
for scalar PDE systems in \Ref{AncWan2020b}. 

\underline{Running example}: 
For the p-gKdV equation \eqref{p-gkdv}, 
the condition \eqref{bracket2.Scondition} holds 
when $Q=Q_3$ where $\ker(S_Q)=\spn(P_1)$ for the symmetry actions $S_1$ and $S_3$,
as seen from Tables~\ref{table:p-gkdv.symmactions} and~\ref{table:p-gkdv.dualsymmactions}. 
The resulting brackets \eqref{adjsymm.skewbracket} and \eqref{adjsymm.antiskewbracket}
on the linear space $\spn(Q_1,Q_2,Q_3)$ 
have the form 
\begin{subequations}
\begin{align}
{}^{Q_3}(Q_1,Q_1)^\pm = {}^{Q_3}(Q_2,Q_2)^\pm = 0,
\quad
{}^{Q_3}(Q_3,Q_3)^+ = Q_3, 
\quad
{}^{Q_3}(Q_3,Q_3)^- = 0,
\\
{}^{Q_3}(Q_1,Q_2)^\pm = 0, 
\quad
{}^{Q_3}(Q_1,Q_3)^\pm = \tfrac{1}{2} \lambda_1^\pm Q_1,
\quad
{}^{Q_3}(Q_2,Q_3)^\pm = \tfrac{1}{2} \lambda_2^\pm Q_2,
\end{align}
\end{subequations}
where, 
for $S_1$: 
$\lambda_1^+=\tfrac{3p-8}{2p-4}$, $\lambda_1^-=\tfrac{p}{4-2p}$, 
$\lambda_2^+=\tfrac{p-8}{2p-4}$, $\lambda_2^-=\tfrac{3p}{4-2p}$;
and for $S_3$:
$\lambda_1^\pm=\pm 1$, $\lambda_2^\pm=\pm 1$.

\subsection{Properties and computational aspects}\label{sec:computational}

As emphasized already, 
the definition of the two brackets \eqref{adjsymm.bracket} and \eqref{adjsymm.bracket2} 
involves the dual linear map $S_Q$ defined by a symmetry action \eqref{S_P.op},
which can be chosen to be any one of the three symmetry actions 
given in Theorem~\ref{thm:symmactions.adjsymm}. 
The different properties of these actions imply corresponding properties for the brackets.

In the case of the second symmetry action \eqref{symmaction1.adjsymm}, 
since it maps any fixed adjoint-symmetry $Q_A$ into a multiplier, 
the resulting brackets will be defined on the linear (sub) space of multipliers 
given by the range of the symmetry action, namely 
$\ran(S_{2\,Q}) \subseteq \multrsp_G \subseteq \adjsymmsp_G$. 
Thus, each of the two brackets will implicitly define a bracket structure
on conservation laws of the given PDE system $G^A=0$
and thereby will constitute a generalized Poisson bracket on 
the conserved integrals associated to the conservation laws. 
Further development of this structure will be left to subsequent work. 

If $Q_A$ is chosen to be a multiplier itself, 
then since the first symmetry action \eqref{symmaction2.adjsymm} coincides 
with the second symmetry action, 
the two brackets defined using the dual linear map $S_{1\,Q}$ 
will be the same as the preceding two brackets. 
Moreover, use of the third symmetry action \eqref{symmaction3.adjsymm} 
when $Q_A$ is a multiplier will produce trivial brackets, 
since $S_{3\,Q}$ vanishes in this case. 

Both of the brackets \eqref{adjsymm.bracket} and \eqref{adjsymm.bracket2} 
are constructed explicitly in terms of $S_Q$ and its inverse $S_Q^{-1}$. 
For the two symmetry actions \eqref{symmaction2.adjsymm} and \eqref{symmaction1.adjsymm}, 
$S_Q$ viewed as an operator involves 
total derivatives $D_I$ and partial derivatives $\partial_{u^\alpha_I}$. 
This means that $S_Q^{-1}$ will involve an integral (operator) with respect to the variables $u^\alpha_I$ in jet space,
whereby the brackets are essentially nonlocal in jet space. 
Nevertheless, as an alternative, 
$S_Q$ can be represented in terms of structure constants that are defined 
with respect to any fixed basis of the linear spaces $\symmsp_G$ and $\adjsymmsp_G$. 
With such a representation, 
the pre-image of any given adjoint-symmetry can be found directly in terms of these structure constants. 
The resulting brackets thus should be viewed as an a posteriori structure 
on the linear space $\adjsymmsp_G$. 

In contrast, for the third symmetry action \eqref{symmaction3.adjsymm}, 
$S_{3\,Q} = Q' + R_Q^*$ is a linear operator in total derivatives, 
where $Q_A$ is any adjoint-symmetry that is not multiplier. 
Consequently, $S_{3\,Q}^{-1}$ only involves the inverse total derivatives $D_I^{-1}$,
and thus the two brackets \eqref{adjsymm.bracket} and \eqref{adjsymm.bracket2} 
are local in jet space and thereby constitute an a priori structure, 
just like the symmetry commutator. 

The same considerations pertain to the corresponding 
pre-symplectic and pre-Hamiltonian (Noether) structures 
shown in Theorem~\ref{thm:symmaction.structures}.

\section{Results for evolution PDEs}\label{sec:evolPDEs}

The preceding general results will next be specialized to evolution PDEs. 

Consider a general system of evolution PDEs for $u^\alpha(t,x)$,
\begin{equation}\label{evol.pde}
u^\alpha_t = g^\alpha(x,u,\partial_x u,\ldots,\partial_x^N u)
\end{equation}
where $x$ now denotes the spatial independent variables $x^i$, $i=1,\ldots,n$,
while $t$ is the time variable. 
In this setting,
the number of PDEs and the number of dependent variables in the system
are equal, $M=m$,
and so the corresponding indices can be identified, $A=\alpha$. 
In particular, 
\begin{equation}
G^\alpha(t,x,u^{(N)}) = u_t^\alpha - g^\alpha(x,u,\partial_x u,\ldots,\partial_x^N u) . 
\end{equation}

It will be useful to note that, 
on the solution space $\Esp$ of the evolution system \eqref{evol.pde}, 
all $t$-derivatives of $u^\alpha$ can be eliminated in any expression
through substituting the equation \eqref{evol.pde} and its spatial derivatives. 
This demonstrates, in particular, that any evolution system 
satisfies Lemma~\ref{lem:Hadamard} and cannot obey any differential identities
\cite{Olv-book,Anc-review}. 
In particular, 
all of the technical conditions assumed in section~\ref{sec:symms.adjointsymms} 
for general PDE systems hold automatically for evolution systems \eqref{evol.pde}. 

The determining equation \eqref{symm.deteqn} for symmetries takes the form 
$(D_tP^\alpha - g'(P)^\alpha)|_\Esp = 0$
for a set of functions $P^\alpha(t,x,u,\partial_x u,\ldots,\partial_x^k u)$
containing no $t$-derivatives of $u^\alpha$. 
The first term can be expressed as
$D_tP^\alpha = \partial_t P^\alpha + P'(u_t)^\alpha
= \partial_t P^\alpha + P'(g)^\alpha +P'(G)^\alpha$,
whence 
\begin{equation}\label{evol.symm.deteqn}
\partial_t P^\alpha + P'(g)^\alpha - g'(P)^\alpha
= \partial_t P^\alpha + [g,P]^\alpha =0
\end{equation}
is the symmetry determining equation in simplified form.
This equation implies that $G'(P)^\alpha = P'(G)^\alpha$ holds off of $\Esp$.
Consequently, one has 
\begin{equation}\label{evol.R_P}
R_P = P' . 
\end{equation}

Likewise,
the determining equation \eqref{adjsymm.deteqn} for adjoint-symmetries 
is given by $(-D_t Q_\alpha - g'{}^*(Q)_\alpha)|_\Esp = 0$
for a set of functions 
$Q_\alpha(t,x,u,\partial_x u,\ldots,\partial_x^k u)$
containing no $t$-derivatives of $u^\alpha$. 
This equation simplifies to the form
\begin{equation}\label{evol.adjsymm.deteqn}
-(\partial_t Q_\alpha +Q'(g)_\alpha + g'{}^*(Q)_\alpha)
=0 . 
\end{equation}
Hence, off of $\Esp$,
one has $G'{}^*(Q)_\alpha = -Q'(G)_\alpha$,
which yields 
\begin{equation}\label{evol.R_Q}
R_Q = -Q' . 
\end{equation}

A useful remark is that the adjoint-symmetry determining equation \eqref{evol.adjsymm.deteqn} can be expressed in the form 
\begin{equation}\label{evol.adjsymm.deteqn.alt}
\partial_t Q_\alpha +\{Q,g\}^*_\alpha=0
\end{equation}
in terms of the anti-commutator $\{A,B\}=A'(B) + B'(A)$,
where $\{A,B\}^*=A'{}^*(B) + B'{}^*(A)$. 
This formulation emphasizes the adjoint relationship 
between the determining equations for adjoint-symmetries and symmetries. 

The necessary and sufficient condition for an adjoint-symmetry to be a conservation law multiplier is that its Frechet derivative is self-adjoint 
\cite{Zha86,FucFok,Olv-book,AncBlu1997,AncBlu2002b,BCA-book,Anc-review}
\begin{equation}\label{evol.multr}
Q'=Q'{}^* . 
\end{equation}
This well-known condition can be expressed more explicitly as the system of Helmholtz equations \cite{AncBlu2002b}
\begin{equation}\label{evol.multr.split}
\partial_{u^\beta_I}Q_\alpha  = (-1)^{|I|} E_{u^\alpha}^{I}(Q_\beta),
\quad
|I|=0,1,\ldots
\end{equation}
in terms of the higher Euler operators $E_{u^\alpha}^{I}$ 
(cf equation \eqref{higher.Euler.op}).
The determining system for multipliers thereby consists of equations \eqref{evol.multr.split} and \eqref{evol.adjsymm.deteqn}.

Self-adjointness \eqref{evol.multr} is also necessary and sufficient 
for $Q_\alpha$ to be a variational derivative (gradient) 
\begin{equation}
\Lambda_\alpha = E_{u^\alpha}(\Phi)
\end{equation}
for some function $\Phi(x,u^{(k)})$, $k\geq 0$. 
Consequently, as is well-known, 
multipliers are variational (gradient) adjoint-symmetries. 

Note that the Frechet derivative identity \eqref{symm.adjsymm.ibp} is given by
\begin{equation}\label{evol.symm.adjsymm.ibp}
Q_\alpha (D_t P^\alpha - g'(P)^\alpha) 
+P^\alpha (D_t Q_\alpha + g'{}^*(Q)_\alpha)
= D_t\Psi^t(P,Q) + D_{x^i} \Psi^i(P,Q)
\end{equation}  
where
\begin{equation}\label{evol.Psi.t}
\Psi^t(P,Q) = Q_\alpha P^\alpha . 
\end{equation}  

\underline{Running example}: 
The low-order adjoint-symmetries \eqref{p-gkdv.Q} of the p-gKdV equation \eqref{p-gkdv}
have the equivalent form 
\begin{equation}\label{p-gkdv.Qevol}
Q_1=u_{xx},
\quad
Q_2=-(u_x{}^p u_{xx} + u_{xxxx}),
\quad
Q_3=2u_x + p x u_{xx} - 3p t (u_x{}^p u_{xx} + u_{xxxx})
\end{equation}
after $t$-derivative of $u$ have been eliminated. 
The first two have the property 
\begin{equation}
Q_1' = D_x^2 = Q_1'{}^*,
\quad
Q_2' = -( p u_x{}^{p-1} u_{xx} D_x + u_x{}^p D_x^2 + D_x^4) 
= -(D_x u_x{}^p D_x +D_x^4) 
= Q_2'{}^* ,
\end{equation}
showing that they are self-adjoint and hence are Euler-Lagrange expressions
\begin{equation}
Q_1 = -\tfrac{1}{2}E_u(u_x{}^2), 
\quad
Q_2 = E_u(\tfrac{1}{(p+1)(p+2)}u_x{}^{p+2} -\tfrac{1}{2}u_{xx}^2) . 
\end{equation}
Correspondingly, they are multipliers. 
The third one satisfies
\begin{equation}\label{g-pkdv.Q3}
Q_3' = 2D_x + p x D_x^2 - 3p t (D_x u_x{}^p D_x + D_x^4)
= Q_3'{}^* + 2(p-2) D_x ,
\end{equation}
showing that it is not self-adjoint and hence is not a multiplier.

\subsection{Symmetry actions on adjoint-symmetries}

The symmetry actions in Theorem~\ref{thm:symmactions.adjsymm} can be simplified
by use of the relations \eqref{evol.R_P} and \eqref{evol.R_Q}.
Combined with the condition \eqref{evol.multr} characterizing multipliers, 
this yields the following result. 

\begin{theorem}\label{thm:evol.symmactions.adjsymm}
The actions \eqref{symmaction2.adjsymm} and \eqref{symmaction1.adjsymm} of symmetries on the linear space of adjoint-symmetries
are respectively given by
\begin{align}
& Q_\alpha \overset{{\X_P}}{\longrightarrow} Q'(P)_\alpha + P'{}^*(Q)_\alpha, 
\label{evol.symmaction2.adjsymm}  
\\
& Q_\alpha \overset{{\X_P}}{\longrightarrow} Q'{}^*(P)_\alpha + P'{}^*(Q)_\alpha ,  
= E_{u^\alpha}(P^\beta Q_\beta) ,
\label{evol.symmaction1.adjsymm}  
\end{align}
which coincide if $Q_\alpha$ is a conservation law multiplier. 
The action \eqref{symmaction3.adjsymm} given by the difference of these two actions consists of 
\begin{equation}\label{evol.symmaction3.adjsymm}
Q_\alpha\overset{{\X_P}}{\longrightarrow} Q'(P)_\alpha -Q'{}^*(P)_\alpha
\end{equation}
which vanishes if $Q_\alpha$ is a conservation law multiplier. 
\end{theorem}

For the sequel, 
indices will be omitted for simplicity of notation wherever it is convenient.

\subsection{Adjoint-symmetry brackets}

For evolution PDEs \eqref{evol.pde}, 
the dual linear map $S_Q$ in the form of 
the adjoint-symmetry commutator bracket \eqref{adjsymm.bracket}
and the non-commutator bracket \eqref{adjsymm.bracket2} 
is given by any of the three symmetry actions in Theorem~\ref{thm:evol.symmactions.adjsymm}. 

Recall that the commutator bracket is well defined 
when $\ker(S_Q)$ satisfies the conditions in either of 
Propositions~\ref{prop:bracket.condition} and~\ref{prop:bracket.condition.alt}. 
The conditions in the first Proposition 
can be expressed entirely in terms of $Q$ and a pair of symmetries $P_1$, $P_2$,
by means of the relations \eqref{evol.R_Q} and \eqref{evol.R_P}. 
In particular, 
condition \eqref{symmaction3.subalg.condition} takes the form 
\begin{equation}\label{evol.symmaction3.subalg.condition}
\pr\X_{P_1}(Q'{}^*)(P_2) -\pr\X_{P_2}(Q'{}^*)(P_1) =0
\end{equation}
while condition \eqref{symmaction1.subalg.condition} takes the form 
\begin{equation}\label{evol.symmaction1.subalg.condition}
\pr\X_{P_1}(Q'{}^*)(P_2) -\pr\X_{P_2}(Q'{}^*)(P_1)
+P_2'{}^*(Q'(P_1)-Q'{}^*(P_1)) - P_1'{}^*(Q'(P_2)-Q'{}^*(P_2))
=0
\end{equation}
for all symmetries 
$\X_{P_1}=P_1^\alpha\partial_{u^\alpha}$ and $\X_{P_2}=P_2^\alpha\partial_{u^\alpha}$
in $\ker(S_Q)$ when $\dim\ker(S_Q)>1$.
When $Q$ is a conservation law multiplier, 
each condition is identically satisfied,
which can be seen from the properties \eqref{evol.multr} and $Q''(P_1,P_2) = Q''(P_2,P_1)$. 

It is worth emphasizing that the existence of these adjoint-symmetry brackets 
does not rely on a PDE system having any variational structure. 
Indeed, 
examples of non-trivial brackets for dissipative PDE systems will be given 
in a subsequent paper.

\subsection{A Noether operator and a symplectic 2-form}

The third symmetry action \eqref{evol.symmaction3.adjsymm} yields 
\begin{equation}\label{evol.Jop}
\Jop=Q'-Q'{}^*
\end{equation}
which is the form of the Noether operator in Theorem~\ref{thm:symmaction.structures}
specialized to evolution PDEs through the relation \eqref{evol.R_Q}.
Note that it will be non-trivial if, and only if, 
$Q$ is a non-variational (non-gradient) adjoint-symmetry. 
This operator is skew, $\Jop^* = -\Jop$. 

From Proposition~\ref{prop:noether.pairing}, 
there is an associated integral bilinear form \eqref{noether.symm.pairing}
on the linear space of symmetries $P^\alpha\partial_{u^\alpha}$. 
Its explicit form for evolution equations is obtained 
by taking the integration domain $\Omega$ to be the spatial domain $\Rnum^n$,  
substituting expression \eqref{evol.Psi.t}, 
and integrating by parts to get 
\begin{equation}\label{evol.2form}
\w_Q(P_1,P_2) = \int_{\Rnum^n} \Psi^t(P_1,\Jop(P_2)) \, d^nx 
= \int_{\Rnum^n} (P_1^\alpha Q'(P_2)_\alpha - P_2^\alpha Q'(P_1)_\alpha)\, d^nx 
\end{equation}
which is manifestly skew. 
Hence, this defines a 2-form on the linear space of symmetries. 
As discussed in Remark~\ref{rem:symplectic.2form}, 
a 2-form is symplectic if it is closed. 
The closure condition, $\d\w_Q =0$, can be formulated as 
\begin{equation}\label{evol.closed.2form}
\pr\X_{f_3}\w_Q(f_1,f_2) +\text{ cyclic } =0
\end{equation}
which must hold for all functions $f_1^\alpha(t,x)$, $f_2^\alpha(t,x)$, $f_3^\alpha(t,x)$. 

\begin{theorem}\label{thm:evol.sys.symplectic.2form}
For any evolution system \eqref{evol.pde}, 
the 2-form \eqref{evol.2form} is symplectic. 
Hence, 
whenever an evolution system admits a 
non-variational (non-gradient) adjoint-symmetry, 
the system possesses a non-trivial associated symplectic structure. 
\end{theorem}

\begin{proof}
Consider 
\begin{equation}
\begin{aligned}
\pr\X_{f_3}\w_Q(f_1,f_2) 
& = \int (f_1^\alpha \pr\X_{f_3}Q'(f_2)_\alpha - f_2^\alpha \pr\X_{f_3}Q'(f_1)_\alpha)\, d^nx
\\
& 
= \int (f_1^\alpha Q''(f_3,f_2)_\alpha - f_2^\alpha Q''(f_3,f_1)_\alpha)\, d^nx . 
\end{aligned}
\end{equation}
Then, in the cyclic sum  
$\pr\X_{f_3}\w_Q(f_1,f_2)  + \pr\X_{f_2}\w_Q(f_3,f_1)  + \pr\X_{f_1}\w_Q(f_2,f_3)$, 
all terms cancel pairwise, due to the symmetry of $Q''$ in its two arguments.
Hence the condition \eqref{evol.closed.2form} is satisfied. 
\end{proof}

The proof can be straightforwardly generalized (using the methods in \Ref{Olv-book}) 
to show that 
\begin{equation}\label{evol.closed.2form.P}
\pr\X_{P_3}\w_Q(P_1,P_2) +\text{ cyclic } =0
\end{equation}
holds for all symmetries $P_1^\alpha\partial_{u^\alpha}$, $P_2^\alpha\partial_{u^\alpha}$, $P_3^\alpha\partial_{u^\alpha}$. 

The formal inverse of the Noether operator \eqref{evol.Jop}
defines a pre-Hamiltonian (inverse Noether) operator $\Jop^{-1}$
which maps adjoint-symmetries into symmetries. 
It also formally yields a Poisson bracket defined by 
\begin{equation}\label{evol.PB}
\{\F_1,\F_2\}_{\Jop^{-1}}:= \int_{\Rnum^n} (\delta \F_1/\delta u)\Jop^{-1}(\delta \F_2/\delta u) \, d^nx
\end{equation}
for functionals $\F=\int_{\Rnum^n} f(x,u^{(k)})\,d^nx$,
where $\delta/\delta u$ denotes the variational derivative, 
namely, $\delta\F/\delta u^\alpha = E_{u^\alpha}(f)$. 

\begin{proposition}\label{prop:evol.PB}
For any non-variational (non-gradient) adjoint-symmetry $Q_\alpha$, 
the bracket \eqref{evol.PB} given by the Noether operator \eqref{evol.Jop} 
is skew and obeys the Jacobi identity 
as a consequence of $\w_Q$ being symplectic. 
\end{proposition}

An interesting general question for future work is 
to determine under what conditions on $\Jop^{-1}$ or $Q_\alpha$ 
will a given evolution equation possessing a Hamiltonian formulation. 

\underline{Running example}: 
The non-gradient adjoint-symmetry 
$Q_3=2u_x + p x u_{xx} - 3p t (u_x{}^p u_{xx} + u_{xxxx})$ 
of the p-gKdV equation \eqref{p-gkdv}
yields the Noether operator (cf \eqref{p-gkdv.noetherop})
$\Jop = Q_3' - Q_3'{}^* = 2(p-2) D_x$
from the relation \eqref{g-pkdv.Q3}. 
Note that the inverse operator acting on $Q_3$ yields
a multiple of the scaling symmetry $P_4=(p-2)u -3p t u_t -xp u_x$. 
The associated symplectic 2-form on $\spn(P_1,P_2,P_3,P_4)$ is explicitly given by 
\begin{equation}\label{p-gkdv.2form}
\w_{Q_3}\big(\sum_{i=1}^{4} a_i P_i, \sum_{j=1}^{4} b_j P_j\big) = 2 \sum_{i,j=1}^{4} a_i b_j \int P_i D_x P_j \, dx . 
\end{equation}
Its components are shown in Table~\ref{table:p-gkdv.2form}. 
Note that, by skew-symmetry, the omitted entries in the lower left 
will be the negative of the entries in the upper right. 
Also observe that the non-zero entries are precisely the 
conserved integrals for momentum \eqref{p-gkdv.momentum} and energy \eqref{p-gkdv.energy}. 
The scaled Noether operator $D_x$ is in fact the inverse of 
the well-known Hamiltonian operator $D_x^{-1}$ 
for which the p-gKdV equation \eqref{p-gkdv} has the Hamiltonian structure
\begin{equation}
u_t = -\tfrac{1}{p+1} u_x{}^{p+1} - u_{xxx} = - D_x^{-1}(\delta H/\delta u),
\quad
H = \int (\tfrac{1}{2}u_{xx}{}^2 - \tfrac{1}{(p+1)(p+2)} u_x{}^{p+2} )\,dx ,
\end{equation}
where $H$ is the energy integral. 

\begin{table}[h!]
\caption{p-gKdV symplectic 2-form} 
\label{table:p-gkdv.2form} 
\begin{tabular}{l||c|c|c|c}
& $P_1$
& $P_2$
& $P_3$
& $P_4$
\\
\hline 
\hline
$P_1$
& $0$
& $0$
& $0$
& $0$
\\
\hline
$P_2$
& 
& $0$
& $0$
& $(4-p)\int  u_x{}^2\, dx$
\\
\hline
$P_3$
& 
& 
& $0$
& $(p+4)\int \big( u_{xx}{}^2 - \tfrac{1}{(p+1)(p+2)}u_x{}^{p+2} \big)\, dx$
\\
\hline
$P_4$
& 
& 
& 
& $0$
\\
\end{tabular}
\end{table}

Further examples of the symplectic structure in Theorem~\ref{thm:evol.sys.symplectic.2form}
will be given in a subsequent paper.

\section{Concluding remarks}\label{sec:remarks}

The work in sections~\ref{sec:symms.adjointsymms} to~\ref{sec:evolPDEs} 
has initiated a mathematical study of the algebraic structure of adjoint-symmetries
for general PDE systems, $G^A(x,u^{(N)})=0$. 
Several main results have been obtained. 

Three linear actions of symmetries on adjoint-symmetries have been derived. 
The first action 
$S_{1\,P}: \adjsymmsp_G\overset{P}{\longrightarrow} \adjsymmsp_G$
comes from applying a symmetry to the determining equation for adjoint-symmetries. 
It yields a generalization of a better known action of
symmetries on conservation law multipliers, 
$\multrsp_G\overset{P}{\longrightarrow} \multrsp_G$. 
The second action arises from a well-known formula that yields 
a conservation law multiplier, $\Lambda_A\in\multrsp_G$, 
from a pair consisting of a symmetry,  $P^\alpha\in\symmsp_G$, and an adjoint-symmetry, $Q_A\in\adjsymmsp_G$. 
Since multipliers are adjoint-symmetries that satisfy certain extra (Helmholtz-type) conditions, 
the formula gives an action 
$S_{1\,P}: \adjsymmsp_G\overset{P}{\longrightarrow} \multrsp_G\subseteq\adjsymmsp_G$. 
A third action $S_{3\,P} := S_{1\,P} -S_{2\,P}$ has the feature that 
it is non-trivial only on adjoint-symmetries that are not multipliers. 

For each of these linear actions, 
two different bilinear brackets on adjoint-symmetries have been constructed
by use of the dual linear action $S_Q(P):=S_P(Q)$ for a fixed adjoint-symmetry. 
The first bracket is a pull-back of the symmetry commutator bracket
and has the properties of a Lie bracket, 
whereas the second bracket does not involve the commutator structure of symmetries 
and is non-symmetric. 
Under certain algebraic conditions on $S_Q$, 
the brackets are well-defined on the entire space of adjoint-symmetries, $\adjsymmsp_G$. 

The third symmetry action is able to produce a Noether (pre-symplectic) operator
whenever a PDE system possesses an adjoint-symmetry that is not a multiplier. 
Furthermore, 
for evolution PDEs, 
this Noether operator gives rise to an associated symplectic 2-form 
which defines a Poisson bracket structure. 
In the case of Hamiltonian systems, 
the Poisson bracket yields an explicit Hamiltonian operator. 

In general, the adjoint-symmetry brackets give 
a correspondence between symmetries and adjoint-symmetries, 
which can exist in the absence of any local variational structure (Hamiltonian or Lagrangian) for a PDE system.  
For the adjoint-symmetry commutator bracket, 
the correspondence constitutes a homomorphism of a Lie (sub) algebra of symmetries 
into a Lie algebra of adjoint-symmetries. 

As shown by the example of the KdV equation in potential form, 
all of these structures are non-trivial, 
which indicates a very rich interplay among conservation laws, adjoint-symmetries and symmetries, 
going beyond the connection provided by Noether's theorem 
and its modern generalization. 
Exploring this interplay more deeply will be an interesting broad aim for future work.

\section*{Acknowledgments}
SCA is supported by an NSERC Discovery Grant and is grateful to Bao Wang 
for useful comments on an earlier version of part of this work. 

This paper is dedicated to George Bluman: 
a wonderful colleague, research collaborator, and book co-author, 
who led me on my first steps into the beautiful subject of symmetries and their applications to differential equations. 
The results bring a full circle to the very fruitful work that we started together on multipliers and adjoint-symmetries more than two decades ago \cite{AncBlu1997}.

\appendix
\section{Calculus in jet space}\label{sec:tools}

General references are provided by \Ref{Olv-book,Anc-review}. 

The following notation is used: 

$x^i$, $i=1,\ldots,n$, are independent variables;

$u^\alpha$, $\alpha=1,\ldots,m$, are dependent variables;

$u^\alpha_i= \frac{\partial u^\alpha}{\partial x^i}$ are partial derivatives; 

$\partial^k u$ is the set of all partial derivatives of $u$ of order $k\geq 0$; 

$u^{(k)}$ is set of all partial derivatives of $u$ with all orders up to $k\geq 0$; 

Multi-indices
$\begin{cases}
I=\emptyset, 
& u^\alpha_I = u^\alpha, 
\quad
|I|=0 
\\
I=\{i_1,\ldots,i_N\}, 
& u^\alpha_I=u^\alpha_{i_1\ldots i_N},
\quad
|I|=N\geq 1
\end{cases}$.

Summation convention:
sum over any repeated (multi-) index in an expression. 

Jet space is the coordinate space $\Jsp=(x^i,u^\alpha,u^\alpha_i,\ldots)$,
and $\Jsp^{(k)} = (x,u^{(k)})$ is the finite subspace of order $k\geq 0$. 

Total derivatives in jet space are defined by 
\begin{equation}
D_i= \partial_{x^i} + u^\alpha_i \partial_{u^\alpha} + \cdots,
\quad
i=1,\ldots,n 
\end{equation}
The Frechet derivative of a function $f$ on jet space is defined by
\begin{equation}
(f')_\alpha = f_{u^\alpha_I}D_I
\end{equation}
which acts on functions $F^\alpha$.
The Frechet second-derivative is given by the expression
\begin{align}
f''(F_1,F_2) = f_{u^\alpha_I u^\beta_J}(D_I F_1^\alpha)(D_J F_2^\beta)
\end{align}
which is symmetric in the pair of functions $(F_1^\alpha,F_2^\alpha)$.
The adjoint of the Frechet derivative of $f$ is defined by
\begin{align}
(f'{}^*)_\alpha = D_I^* f_{u^\alpha_I}= (-1)^{|I|} D_I f_{u^\alpha_I}
\end{align}
which acts on functions $F$, 
where the righthand side is a composition of operators. 

The Euler operator (variational derivative) is defined by 
\begin{equation}\label{Euler.op}
  E_{u^\alpha} = (-1)^{|I|} D_I\partial_{u^\alpha_I}
\end{equation}  
It has the property that $E_{u^\alpha}(f) =0$ holds identically iff $f= D_i F^i$
for some vector function $F^i(x,u^{(k)})$.
The product rule for the Euler operator is given by 
\begin{equation}
E_{u^\alpha}(f_1f_2) = f_1'{}^*(f_2)_\alpha + f_2'{}^*(f_1)_\alpha
\end{equation} 

The higher Euler operators are defined similarly 
\begin{equation}\label{higher.Euler.op}
E_{u^\alpha}^{I} = \smallbinom{I}{J} (-1)^{|I/J|} D_{I/J}\partial_{u^\alpha_{I}}
\end{equation}    
See \Ref{Olv-book,Anc-review} for their properties. 

Some useful relations: 
\begin{gather}
f'(F) = F^\alpha E_{u^\alpha}(f) + D_i\Gamma^i(F;f),
\quad
\Gamma^i(F;f)= (D_I F^\alpha) E_{u^\alpha_{iI}}(f);
\\
H f'(F) -F f'{}^*(H) = D_i\Psi^i(H,F),
\quad
\Psi^i(H,F)= (D_I H)(D_J F^\alpha) (-1)^{|I|} E_{u^\alpha_J}^{I}(f);
\\
f'(F) =  \pr\X_F f, 
\quad
\X_F = F^\alpha\partial_{u^\alpha},
\quad
\pr\X_F = (D_I F^\alpha)\partial_{u_I^\alpha} ;
\\
[F_1,F_2] = \pr\X_{F_1}F_2 -\pr\X_{F_2}F_1 = F_2'(F_1) -F_1'(F_2) ;
\\
[\pr\X_{F_1},\pr\X_{F_2}] = \pr\X_{[F_1,F_2]} ;
\end{gather}
and
\begin{align}
(\pr\X_F f') & = (\pr\X_F f)' - f' F' ;
\label{X.Frechet.id}
\\
(\pr\X_F f'{}^*) & = (\pr\X_F f)'{}^* - F'{}^* f'{}^* . 
\label{X.adjFrechet.id}
\end{align}

\end{document}